\documentclass[11pt,twoside, letter]{article}
\usepackage{amsthm, amsmath, amssymb, amsfonts, graphicx, epsfig}
\usepackage{accents}

\usepackage{bbm}
\usepackage{mathrsfs}
\usepackage{dsfont}

\usepackage{epstopdf}

\usepackage{graphicx}
\usepackage[font=sl,labelfont=bf]{caption}
\usepackage{subcaption}

\usepackage{float}
\restylefloat{table}

\usepackage{enumerate}

\usepackage[usenames,dvipsnames]{color}

\usepackage[colorlinks=true, pdfstartview=FitV, linkcolor=blue,
            citecolor=blue, urlcolor=blue]{hyperref}

\usepackage{url}
\makeatletter\def\url@leostyle{%
 \@ifundefined{selectfont}{\def\UrlFont{\sf}}{\def\UrlFont{\scriptsize\ttfamily}}} \makeatother

\usepackage[margin=1.0in, letterpaper]{geometry}

\newtheorem{theorem}{Theorem}
\newtheorem{proposition}[theorem]{Proposition}
\newtheorem{lemma}[theorem]{Lemma}
\newtheorem{corollary}[theorem]{Corollary}

\theoremstyle{definition}
\newtheorem{definition}[theorem]{Definition}
\newtheorem{example}[theorem]{Example}
\theoremstyle{remark}
\newtheorem{remark}[theorem]{Remark}

\theoremstyle{Assumption}
\newtheorem{assumption}[theorem]{Assumption}

\numberwithin{equation}{section}
\numberwithin{theorem}{section}

\definecolor{Red}{rgb}{0.9,0,0.0}
\definecolor{Blue}{rgb}{0,0.0,1.0}

\def\cB{\mathcal{B}}
\def\cC{\mathcal{C}}
\def\cD{\mathcal{D}}

\def\cG{\mathcal{G}}
\def\cH{\mathcal{H}}

\def\cM{\mathcal{M}}
\def\cN{\mathcal{N}}

\def\cS{\mathcal{S}}

\def\cX{\mathcal{X}}
\def\cY{\mathcal{Y}}

\def\bG{\mathbb{G}}

\def\bP{\mathbb{P}}
\def\bQ{\mathbb{Q}}
\def\bR{\mathbb{R}}

\def\sC{\mathscr{C}}

\renewcommand{\mid}{\;|\;}              

\DeclareMathOperator*{\esssup}{ess\,sup} 
\DeclareMathOperator*{\essinf}{ess\,inf} 

\def\I{\mathds{1}}
\def\wh{\widehat}
\def\wt{\widetilde}
\def\phi{\varphi }

\newcommand{\Vcom}{{V}^{\textrm{com}}}

\newcommand{\Vcomt}{{\widetilde{V}}^{\textrm{com}}}

\def\FVA{{\rm FVA}}
\def\CVA{{\rm CVA}}
\def\DVA{{\rm DVA}}

\def\C-FVA{{\rm C-FVA}}

\def\XVA{\textrm{XVA}}

\def\DVA{\textrm{DVA}}

\def\LVA{\textrm{LVA}}

\newcommand{\undp}{{\underline p}}
\newcommand{\ovep}{{\overline p}}
\newcommand{\undpa}{{\underline p}^a}
\newcommand{\ovepl}{{\overline p}^l}
\def\whtau{\widehat{\tau }}

\usepackage{multirow}
\usepackage{titling}
\title{  Arbitrage-Free Pricing Of Derivatives In Nonlinear Market Models }

\makeatletter
\def\and{%
  \end{tabular}%
  \begin{tabular}[t]{c}}%
\def\@fnsymbol#1{\ensuremath{\ifcase#1\or a\or b\or c\or
   d\or e\or f\or g\or h\or i\else\@ctrerr\fi}}
\makeatother

\author{
        Tomasz R. Bielecki\,\thanks{Department of Applied Mathematics, Illinois Institute of Technology
       \newline \hspace*{1.45em}  10 W 32nd Str, Building E1, Room 208, Chicago, IL 60616, USA
       \newline \hspace*{1.45em}  Emails: \url{tbielecki@iit.edu} (T.R. Bielecki) and \url{cialenco@iit.edu} (I. Cialenco)
       \newline \hspace*{1.45em}  URLs: \url{http://math.iit.edu/\~bielecki}  and \url{http://math.iit.edu/\~igor}
        \vspace{0.5em}} ,
\and
        Igor Cialenco\,\footnotemark[1] ,
\and and
Marek Rutkowski\,\thanks{School of Mathematics and Statistics, University of Sydney, Sydney, NSW 2006, Australia
 \newline \hspace*{1.45em} and Faculty of Mathematics and Information Science, Warsaw University of Technology, 00-661 Warszawa, Poland
        \newline \hspace*{1.45em}    Email: \url{marek.rutkowski@sydney.edu.au}, URL: \url{http://sydney.edu.au/science/people/marek.rutkowski.php} }
        }

\date{}
\date{
First Circulated: January 28, 2017
This Version: April 4, 2018 \\[1em]
Forthcoming in \textit{Probability, Uncertainty and Quantitative Risk}
}

\begin{document}

\maketitle


{\footnotesize
\begin{tabular}{l@{} p{350pt}}
  \hline \\[-.2em]
  \textsc{Abstract}: \ &
  The objective of this paper is to provide a comprehensive study no-arbitrage pricing of financial derivatives in the presence of funding costs, the counterparty credit risk and market frictions affecting the trading mechanism, such as collateralization and capital requirements. To achieve our goals, we extend in several respects the nonlinear pricing approach developed in El Karoui and Quenez \cite{EQ1997} and El Karoui et al. \cite{EPQ1997},
which was subsequently continued in Bielecki and Rutkowski \cite{BR2015}. \\[0.5em]
\textsc{Keywords:} \ &  hedging, fair price, funding cost, margin agreement, market friction, BSDE \\
\textsc{MSC2010:} \ & 91G40, 60J28  \\[1em]
  \hline
\end{tabular}
}

\tableofcontents
\newpage

\section{Introduction}   \label{sec1}

The paper contributes to the {\it nonlinear} arbitrage-free pricing theory, which arises in a natural way due to the salient features of real-world trades, such as: trading constraints, differential funding costs, collateralization, counterparty credit risk, and capital requirements. Our aim is to extend in several respects the nonlinear hedging and pricing approach developed in El Karoui and Quenez \cite{EQ1997} and El Karoui et al. \cite{EPQ1997} who used a BSDE approach, by accounting for the complexity of over-the-counter financial derivatives and specific features of the trading environment after the global financial crisis. This work builds also upon the earlier paper by Bielecki and Rutkowski \cite{BR2015} where, however, the
important issue of no-arbitrage was not studied in depth. The paper is structured as follows:
 \vskip - 12 pt \begin{itemize} \vskip - 12 pt  {\parskip = 0 pt
\item In Section \ref{sec2}, we introduce self-financing trading strategies in the presence of differential funding rates and adjustment processes. We consider general contracts with cash flow streams, rather than simple contingent claims with a single payoff either at the contract's maturity or upon early exercise. We also introduce in Section \ref{sec2n1} the concepts of {\it local} and {\it global} valuation problems. This distinction is crucial since it demonstrates that results obtained in Sections \ref{sec3} and \ref{sec3.7} are capable of covering also financial models and valuation problems that cannot be addressed through classical BSDEs, which are nowadays commonly used to deal with nonlinear financial markets.
\item Section \ref{sec3} is devoted to a comprehensive examination of the issue of existence of arbitrage opportunities for the hedger and for the trading desk in a nonlinear trading framework and with respect to a predetermined class of contracts. We introduce the concept of no-arbitrage with respect to the null contract and a stronger notion of no-arbitrage for the trading desk. We then proceed to the issue of unilateral fair valuation of a given contract by the hedger who is endowed with an initial capital. We examine the link between the concept of no-arbitrage for the trading desk and the financial viability of prices computed by the hedger.
\item In Section \ref{sec3.7}, we propose and analyze the concept of a {\it regular} market model, which can be seen as an extension of the notion of a {\it nonlinear pricing system,} which was introduced by El Karoui and Quenez \cite{EQ1997}.
    The goal is to identify a class of nonlinear markets, which are arbitrage-free for the trading desk and, in addition, enjoy the desirable property that if a given contract can be replicated, then the cost of replication is also the fair price for the hedger.
\item Section \ref{sec4} focuses on replication of a contract in a regular market model. We propose four alternative definitions of no-arbitrage prices, namely, the {\it gained value}, the {\it ex-dividend price}, the {\it exit price}, and the {\it offsetting price}. Generally speaking, it is not expected that these prices will coincide, since they correspond to different valuation problems for the hedger. However, when the trading arrangements in the underlying model are such that the valuation problem is \textit{local} then, under some suitable technical conditions, we show that the gained value and the ex-dividend price coincide.
\item In Section \ref{sec6}, we present a BSDEs approach to the valuation and hedging and we give examples of BSDEs for the gained value and the ex-dividend price. Finally, we briefly address in Section \ref{sec7} the issue of {\it valuation
     adjustments} in linear and nonlinear markets and we make some comments on the prevailing market practice of performing separate computations of the so-called `clean price' and the `total valuation adjustment', and subsequently adding them to obtain the full price charged to customers.}
\end{itemize}

Although we focus on the issue of fair unilateral valuation from the perspective of the hedger, it is clear that identical definitions and valuation methods are applicable to his counterparty as well. Hence, in principle, it is possible to use our results to examine the interval of fair bilateral prices in a regular market model. Particular instances of such unilateral and bilateral valuation problems were previously studied in Nie and Rutkowski \cite{NR2015,NR2016a,NR2017} where it was shown that a non-empty interval of either fair bilateral prices or bilaterally profitable prices can be obtained in some nonlinear markets for contracts with either an exogenous or an endogenous collateralization.  It should be acknowledged that there exists a vast body of literature devoted to valuation and hedging of financial derivatives under differential funding costs, collateralization, the counterparty credit risk and other trading adjustments (see, for instance, Bichuch et al. \cite{BCS2017}, Brigo and Pallavicini \cite{BP2014}, Brigo et al. \cite{BB2018,BBR2017}, Burgard and Kjaer \cite{BK2011,BK2013}, Cr\'epey \cite{C2015a,C2015b},  Mercurio \cite{M2013}, Pallavicini et al. \cite{PPB2012b,PPB2012a}, and Piterbarg \cite{P2010}). In view of limited space, we cannot present here these works in detail. Let us only mention that most of these papers deal with {\bf linear} markets with credit risk (possibly also with differential funding rates), whereas the general theory developed in this work aims to address problems where the emphasis is put on a {\bf nonlinear} character of valuation in markets with imperfections. In contrast,  Albanese et al. \cite{ACC2017}, Albanese and Cr\'epey \cite{AC2017}, and Cr\'epey et al. \cite{BES2017} propose to address the issue of valuation adjustments through an alternative approach, which is based on the global valuation paradigm referencing to the balance sheet of the bank, its internal structure, and long-term interests of bank's shareholders. The issue of nonlinearity of trading does not appear in their approach, since the classical hedging arguments are no longer employed to determine the value of a new contract, which is added to the existing portfolio of bank's assets. For further comments on some of the above-mentioned papers, we refer to Section~\ref{sec7}.

\section{Nonlinear Market Model}   \label{sec2}

We start by re-examining and extending the nonlinear trading setup introduced in Bielecki and Rutkowski \cite{BR2015}. Throughout the paper, we fix a finite trading horizon date $T>0$ for our market model. Let $(\Omega, \cG, \bG , \bP)$ be a filtered probability space satisfying the usual conditions of right-continuity and completeness, where the filtration $\bG=(\cG_t)_{t \in [0,T]}$ models the flow of information available to the hedger and his counterparty. For convenience, we assume that the initial $\sigma$-field ${\cal G}_0$ is trivial. All processes introduced in what follows are implicitly assumed to be $\bG$-adapted and, as usual, any semimartingale is assumed to be a c\`adl\`ag process. Let us introduce the notation for the prices of all traded assets in our model.

\vskip 5 pt
\noindent {\bf Risky assets.} We denote by $\mathcal{S}=(S^1,\ldots,S^d)$ the collection of the {\it ex-dividend prices} of a family of $d$ risky assets with the corresponding {\it cumulative dividend streams} $\mathcal{D} =(D^1,\ldots, D^d)$. The process $S^i$ represents the ex-dividend price of any traded security, such as, stock, sovereign or corporate bond, stock option, interest rate swap, currency option or swap, credit default swap, etc.

\vskip 5 pt
\noindent {\bf Funding accounts.}  We denote by $B^{i,l}$ (resp. $B^{i,b}$) the {\it lending} (resp. {\it borrowing}) {\it funding account} associated with the $i$th risky asset, for $i=1,2,\ldots,d$. The financial interpretation of these accounts varies from case to case. For an overview of trading mechanisms for risky assets, we refer to Section \ref{sec2.6}. In the special case when $B^{i,l}=B^{i,b}$, we will use the notation $B^i$ and we call it the {\it funding account} for the $i$th risky asset.

\vskip 5 pt
\noindent {\bf Cash accounts.} The \textit{lending cash account} $B^{0,l}$  and the {\it borrowing cash account} $B^{0,b}$
are used for unsecured lending and borrowing of cash, respectively. For brevity, we will sometimes write $B^{l}$ and $B^{b}$
instead of $B^{0,l}$ and $B^{0,b}$. Also, when the borrowing and lending cash rates are equal, the single cash
account is denoted by $B^{0}$ or, simply, $B$. Note, however, that since an unlimited borrowing/depositing of cash in the bank account is not a realistic feature of a trading model, it is not assumed in what follows.

For brevity, we denote by $\mathcal{B} =(B^{i,l},B^{i,b},\ i=0,1,\ldots,d)$ the collection of all cash and funding accounts.

\subsection{Contracts with Trading Adjustments}   \label{sec2.1}

We will consider financial contracts between two parties, called the \textit{hedger} and the \textit{counterparty}. In what follows, all the cash flows will be viewed from the prospective of the hedger, with the convention that a positive cash flow means that the hedger receives the corresponding amount, and a negative cash flow meaning that the hedger makes a payment. A \textit{bilateral financial contract} (or simply a {\it contract}) is given as a pair $\cC=(A,\cX)$ where the meaning of each term is explained below.

A stochastic processes $A$ represents the {\it cumulative cash flows} from time 0 till the contracts's maturity date, which is denoted as $T$. In the financial interpretation, the process $A$ is assumed to model the cumulative (promised) cash flows of a given contract,  which are either paid out from the hedger's wealth or added to his wealth via the value process of his portfolio of traded assets (including positive or negative holdings of cash, that is, lent or borrowed money). Note that the price of the contract $\cC$ exchanged at its initiation (that is, at time 0) is not included in $A$. For example, if the contract stipulates that the hedger will `receive' the (possibly random and of arbitrary signs) cash flows  $a_1,a_2,\dots,a_m$ at times $t_1,t_2,\dots,t_m \in (0,T]$, then $A$ is given by
\begin{equation*}
A_t=\sum_{l=1}^m \I_{[t_l,\infty )}(t) a_l .
\end{equation*}
Let $(A^t,0)$ denote a {\it basic contract} originated at time $t$ with $\cX =0$. Then the only cash flow exchanged between the counterparties at time $t$ is the price of the contract and thus the remaining cumulative cash flows of $(A^t,0)$ are given as $A^t_u := A_u-A_t$ for $u \in [t,T]$. In particular, the equality $A^t_t=0$ is valid for any basic contract $(A,0)$ and any date $t \in [0,T)$. All future cash flows $a_l$ for $l$ such that $t_l >t$ are predetermined, in the sense that they are explicitly specified by the contract covenants.

As a simple example of cash flows, consider the situation where the hedger sells at time $t$ the European call option on the risky asset $S^i$. Then $m=1,\, t_1=T$, and the terminal payoff from the perspective of the hedger equals $a_1=- (S^i_T-K)^+$. More generally, for every $t \in [0,T)$, the process $A^t$ is given by $A^t_u=- (S^i_T-K)^+ \I_{[T,\infty )}(u)$ for every $u \in [t,T]$.

To account for additional features of a particular contract at hand, we find it convenient to postulate that the cash flows $A$ (resp. $A^t$) of a basic contract are complemented by {\it trading adjustments}, which are represented by the process $\mathcal{X}$  (resp. $\mathcal{X}^t$) given as $\mathcal{X}=(X^{1},\ldots, X^{n}; \alpha^1, \ldots, \alpha^n; \beta^1, \ldots, \beta^n).$ The role of $\mathcal{X}$ is to describe additional clauses of a given contract, such as rehypothecated or segregated collateral, as well as to account for the impact of atypical trading arrangements on the value process of the hedger's portfolio. For each {\it adjustment process} $X^k$, the process $\alpha^k X^k$ represents additional incoming or outgoing cash flows for the hedger, which are either stipulated in the clauses of the contract or imposed by a third party (for instance, the regulator). To each process $X^k, \ k=1,2,\ldots,n$ we also specify the {\it remuneration process} $\beta^{k}$, which is used to determine the net interest payments (if any) associated with the process $X^k$. It should be noted that the processes $X^1, \dots , X^n$ and the associated remuneration processes $\beta^1, \dots , \beta^n$ do not represent traded assets, although they impact the dynamics of the value of a portfolio (see \eqref{eq:gains}). It is rather clear that the processes $\alpha$ and $\beta$ may depend on the respective adjustment process. Therefore, when the adjustment process is $\cY$, rather than $\cX$, one should write  $\alpha(\cY)$ and $\beta(\cY)$ in order to avoid confusion. However, for brevity, we will keep the shorthand notation $\alpha$ and $\beta$ when the adjustment process is denoted as $\cX$.   For further comments on trading adjustment, we refer to Section \ref{sec2.3} and \ref{sec2.4}. Last, but not least, we will need to define a suitable modification of the promised cash flows $A$ resulting from the counterparty credit risk; see Definition  \ref{close}  where the concept of {\it counterparty risky cumulative cash flows} is introduced.

In essence, the unilateral valuation of a given contract is the process of finding at any date $t$ the range of its {\it fair prices} $p_t$, as seen from the viewpoint of either the hedger or the counterparty. Although it will be postulated that the two parties in a contract adopt the same valuation paradigm, due to the asymmetry of cash flows, differential trading costs, and possibly also different trading opportunities, they will typically obtain different ranges for the respective fair unilateral prices of a bilateral contract. The disparity in unilateral valuation executed independently by the two parties is a consequence of the nonlinearity of the wealth dynamics in trading strategies, so that it will typically occur even within the framework of a complete nonlinear market, where the perfect replication of a contract can be achieved by the counterparties. An important issue of determination of the range of fair bilateral prices in a general nonlinear framework is left for a future work; for results on bilateral pricing in specific nonlinear market models, see Nie and Rutkowski \cite{NR2015,NR2016a,NR2017}.

\subsection{Self-financing Trading Strategies}    \label{sec2.2}

The concept of a portfolio refers to the family of {\it primary traded assets}, that is, risky assets, cash accounts, and funding accounts for risky assets.  Formally, by a {\it portfolio} on the time interval $[t,T]$, we mean an arbitrary ${\mathbb R}^{3d+2}$-valued, $\bG$-adapted process $(\phi^t_u)_{u \in [t,T]}$ denoted as
\begin{equation} \label{por1}
\phi^t=\big( \xi^1,\ldots , \xi^{d}; \psi^{0,l},\psi^{0,b},\psi^{1,l}, \psi^{1,b}, \dots ,\psi^{d,l}, \psi^{d,b}  \big),
\end{equation}
where the components represent the positions in risky assets $(S^i,D^i),\, i=1,2, \dots , d$, cash accounts $B^{0,l},\, B^{0,b}$, and funding accounts $B^{i,l} ,B^{i,b} ,\, i=1,2, \dots , d$ for risky assets. It is postulated throughout that $\psi^{j,l}_u \geq 0,\, \psi^{j,b}_u \leq 0$ and $\psi^{j,l}_u \psi^{j,b}_u =0$ for all $j=0,1, \dots ,d$ and $u \in [t,T]$. If the borrowing and lending rates are equal, then we write $\psi^j =\psi^{j,l}+\psi^{j,b}$. It is also assumed throughout that the processes $\xi^1,\ldots , \xi^{d}$ are $\bG$-predictable.

We say that a portfolio $\phi$ is {\it constrained} if at least one of the components of the process $\phi$ is assumed
to satisfy some explicitly stated constraints, which directly affect the choice of $\phi $. For instance, we will need to impose conditions ensuring that the funding of each risky asset is done using the corresponding funding account. Another example of an explicit constraint is obtained when we set $\psi^{0,b}_u =0$ for all $u \in [t,T]$, meaning that an outright borrowing of cash from the account $B^{0,b}$ is prohibited. For examples of markets with various kinds of portfolio constraints, we refer to Carassus et al. \cite{CPT2001}, Fahim and Huang \cite{FH2016}, Karatzas and Kou \cite{KK1996,KK1998}, and Pulido \cite{P2014} and the references therein. The concept of a constrained portfolio should be contrasted with the notion of {\it admissibility} of a trading strategy that may involve some additional conditions imposed on the wealth process and thus indirectly also on the class of admissible processes $\phi $ (see Definition \ref{defadmih}). Note that portfolio constraints are not a matter of choice, since they are due to genuine real-life restrictions imposed on traders. This should be contrasted with the idea of {\it admissibility} of a trading strategy, which is a mathematical artifact needed to preclude unrealistic arbitrage opportunities (like doubling strategies), which may be present within a stochastic model when continuous trading is allowed. Note in this regard that there is no need to be concerned with the admissibility under the realistic assumption that only a finite number of trading times is available to traders.

We are now in a position to state some standard technical assumptions underpinning our further developments.

\begin{assumption} \label{assumpassets}
{\rm We work throughout under the following standing assumptions:
\begin{enumerate}[(i)]\addtolength{\itemsep}{-.5\baselineskip}
\item  for every $i=1,2,\dots , d$, the price $S^i$ of the $i$th risky asset is a semimartingale and the cumulative dividend stream $D^i$ is a process of finite variation with $D^i_{0}=0$;
\item the cash and funding accounts $B^{j,l}$ and $B^{j,b}$  are strictly positive and continuous processes of finite variation
       with $B^{j,l}_0=B^{j,b}_0=1$ for $j=0,1,\ldots,d$;
\item the cumulative cash flow process $A$ of any contract is a process of finite variation;
\item the adjustment processes $X^k, \, k=1,2,\ldots, n$ and the auxiliary processes $\alpha^k,  \, k=1,2,\ldots, n$ are semimartingales;
\item the remuneration processes $\beta^k, \, k=1,2,\ldots, n $  are strictly positive and continuous processes of finite variation
       with $\beta^k_0=1$ for every $k$.
\end{enumerate} }
\end{assumption}

In the next definition, the $\cG_t$-measurable random variable $x_t$ represents the {\it endowment} of the hedger at time $t \in [0,T)$ whereas $p_t$, which at this stage is an arbitrary $\cG_t$-measurable random variable, stands for the price at time $t$ of $\cC^t=(A^t ,\cX^t)$, as seen by the hedger. Recall that $A^t$ denotes the cumulative cash flows of the contract $A$ that occur after time $t$, that is, $A^t_u:=A_u-A_t$ for all $u \in [t,T]$. Hence $A^t$ can be seen as a contract with the same remaining cash flows as the original contract $A$, except that $A^t$ starts and is traded at time $t$. By the same token, we denote by $\cX^t$ the adjustment process related to the contract $A^t$. Let $\sC$ be a predetermined class of contracts. As expected, it is assumed throughout that the null contract $\cN=(0,0)$ is traded in any market model at any time $t$, that is, $\cN \in \sC^t$ for every $t \in [0,T)$ (see Assumption~\ref{ass3.1}).

It should be noted that the prices $p_t$ for contracts belonging to the class $\sC$ are yet unspecified and thus there is a certain degree of freedom in the foregoing definitions. Note also that we use the convention that $\int_t^u:=\int_{(t,u]}$ for any $t \leq u$.

\begin{definition} \label{ts2}
A  quadruplet  $(x_t,p_t,\phi^t,\cC^t)$ is a \textit{self-financing trading strategy on} $[t,T]$ associated with the contract $\cC=(A,\cX)$ if the {\it portfolio value} $V^p (x_t,p_t,\phi^t,\cC^t)$, which is given by
\begin{equation} \label{wea1}
V^p_u (x_t,p_t,\phi^t,\cC^t):= \sum_{i=1}^{d} \xi^i_u S^i_u+\sum_{j=0}^{d} \left( \psi^{j,l}_u B^{j,l}_u+\psi_u^{j,b} B^{j,b}_u \right)
\end{equation}
satisfies for all $u \in [t,T]$
\begin{equation} \label{eq:selfFin1}
V^p_u (x_t,p_t,\phi^t,\cC^t)=x_t +p_t+ G_u(x_t,p_t,\phi^t,\cC^t),
\end{equation}
where the \textit{adjusted gains process} $G(x_t,p_t,\phi^t,\cC^t)$ is given by
\begin{equation}  \label{eq:gains}
\begin{aligned}
G_u (x_t,p_t,\phi^t,\cC^t) :=  &  \sum_{i=1}^{d}  \int_t^u \xi^i_v \, (dS^i_v+dD^i_v )+ \sum_{j=0}^{d}  \int_t^u \left( \psi^{j,l}_v\,dB^{j,l}_v+\psi_v^{j,b}\, dB^{j,b}_v \right) \\
&+\sum_{k=1}^n \alpha^k_u X^k_u-\sum_{k=1}^{n}\int_t^u X^k_v (\beta^{k}_v)^{-1}\,d\beta^{k}_v+A_u^t.
\end{aligned}
\end{equation}
For a given pair $(x_t,p_t)$, we denote by $\Phi^{t,x_t}(p_t,\cC^t)$ the set of all self-financing trading strategies on $[t,T]$ associated with the contract $\cC$.
\end{definition}

When studying the valuation of a contract $\cC^t$ for a fixed $t$, we will typically assume that the hedger's endowment $x_t$ is given and we will search for the range of hedger's fair prices $p_t$ for $\cC^t$. Therefore, when dealing with the hedger with a fixed initial endowment $x_t$ at time $t$, we will consider the following set of self-financing trading strategies $\Phi^{t,x_t}(\sC )=\cup_{\cC  \in \sC } \cup_{p_t \in \cG_t}\Phi^{t,x_t} (p_t,\cC^t)$. Note, however, that the definition of the market model does not assume that the quantity $x_t$ is predetermined.

\begin{definition}\label{def:UnderMarket}
The \textit{market model} is the quintuplet $\cM=(\cS,\cD, \cB, \sC ,\Phi(\sC))$ where $ \Phi(\sC)$ stands for the set of
all self-financing trading strategies associated with the class $\sC $ of contracts, that is, $\Phi (\sC )= \cup_{t \in [0,T)} \cup_{x_t \in \cG_t} \Phi^{t,x_t}( \sC ).$
\end{definition}

In principle, the market model defined above exhibits nonlinear features, in the sense that either the portfolio value process $V^p(x_t,p_t,\phi^t,\cC^t)$ is not linear in $(x_t, p_t,\phi^t,\cC^t)$ or the class of all self-financing strategies is not a vector space (or, typically, both). Therefore, we refer to this setup as to a generic {\it nonlinear market model}. In contrast, by a {\it linear market model} we will understand in this paper the version of the model defined above in which all trading adjustments are null (i.e., $X^k=0$ for all $k=1,2, \dots , n$), there are no differential funding rates (i.e., $B^{j,b}=B^{j,l}$ for all $j=0,1, \dots, d$) and no portfolio constraints are imposed. In particular, in the linear market model the class of all self-financing trading strategies is a vector space and the value process $V^p(x_t,p_t,\phi^t,\cC^t)$ is a linear mapping in $(x_t,p_t,\phi^t,\cC^t)$. Note, however, that the last property is usually lost when an admissibility condition is imposed on the class of trading strategies since, typically, a trading strategy is deemed to be {\it admissible} if it its discounted wealth is bounded from below or nonnegative (hence the class of admissible trading strategies is no longer a vector space).

To alleviate notation, we will frequently write $(x,p,\phi,\cC)$ instead of $(x_0,p_0,\phi^0,\cC^0)$ when working on the interval $[0,T]$.
Note that \eqref{wea1}--\eqref{eq:gains} yield the following equalities for any trading strategy $(x, p ,\phi,\cC) \in \Phi^{0,x}(\sC)$
\begin{equation} \label{xsxs}
V^p_0(x,p,\phi,\cC)=\sum_{i=1}^d \xi_0^i S_0^i+\sum_{j=0}^d \left(\psi_0^{j,l}B_0^{j,l}+\psi_0^{j,b}B_0^{j,b}\right)
=x+p+\sum_{k=1}^{n}\alpha^k_0 X_0^k.
\end{equation}
Recall that in the classical case of a frictionless market, it is common to assume that the initial endowments of traders are null. Moreover, the price of a derivative has no impact on the dynamics of the gains process. In contrast, when portfolio's value is driven by nonlinear dynamics, the initial endowment $x$ at time $0$, the initial price $p$ and the adjustment cash flows of a contract may all affect the dynamics of the gains process and thus the classical approach is no longer valid.

\subsection{Funding Adjustment}           \label{sec2.3}

The concept of the funding adjustment refers to the spreads of funding rates with regard to some benchmark cash rate.
In the present setup, it can be defined relative to either $B^{l}$ or $B^{b}$. If the lending and borrowing rates are not equal,
then \eqref{eq:selfFin1} can be written as follows
\begin{align*}  
V_t^p(x,&p,\phi,\cC)= x +p+\sum_{i=1}^{d}  \int_0^t \xi^i_u \, (dS^i_u+dD^i_u )+\sum_{k=1}^n \alpha^k_t X^k_t+ A_t  \nonumber \\
&+\sum_{j=0}^d \int_0^t  \left( \psi_u^{j,l}\, dB_u^{0,l}+\psi_u^{j,b}\, dB_u^{0,b} \right) \nonumber \\
&-\sum_{k=1}^n \int_0^t \bigg( (X_u^{k})^{+}(B_u^{0,l})^{-1}\,dB_u^{0,l}-(X_u^{k})^{-}(B_u^{0,b})^{-1}\, dB_u^{0,b} \bigg) \nonumber \\
&+\sum_{i=1}^d\int_0^t \left( \psi_u^{i,l}\big((\widehat B^{i,l}_u-1)\, dB_u^{0,l}+B^{0,l}_u\,d\widehat B_u^{i,l} \big)
+\psi_u^{i,b}\big((\widehat B^{i,b}_u-1)\,dB_u^{0,b}+B^{0,b}_u\,d\widehat B_u^{i,b} \big) \right)    \\
&- \sum_{k=1}^{n}  \int_0^t \left( (X_u^{k})^{+} (\widehat\beta_u^{k,l})^{-1}\, d\widehat \beta_u^{k,l}-(X_u^{k})^{-} (\widehat\beta_u^{k,b})^{-1}\, d\widehat \beta_u^{k,b} \right) \nonumber
\end{align*}
where $\widehat B^{j,l/b} := B^{j,l/b} (B^{0,l/b})^{-1}$ and $\widehat \beta^{k,l/b} := \beta^k(B^{0,l/b})^{-1} $.
The quantity
\begin{align*}
\gamma_t&:= \sum_{i=1}^d\int_0^t \left( \psi_u^{i,l}\big((\widehat B^{i,l}_u-1)\, dB_u^{0,l}+B^{0,l}_u\,d\widehat B_u^{i,l} \big)
+\psi_u^{i,b}\big((\widehat B^{i,b}_u-1)\, dB_u^{0,b}+B^{0,b}_u\,d\widehat B_u^{i,b} \big) \right)   \\
&- \sum_{k=1}^{n}  \int_0^t \left( (X_u^{k})^{+} (\widehat\beta_u^{k,l})^{-1}\, d\widehat \beta_u^{k,l}-(X_u^{k})^{-} (\widehat\beta_u^{k,b})^{-1}\, d\widehat \beta_u^{k,b} \right)
\end{align*}
is called the \textit{funding adjustment}. If the borrowing and lending rates are equal, then the expression for the funding adjustment simplifies to
\[
\gamma_t = \sum_{i=1}^{d}  \int_0^t\psi^i_u \big((\wh B^i_u-1)\, dB^0_u+ B_u^0\, d\wh B^i_u \big) -
\sum_{k=1}^{n}  \int_0^t X_u^k(\wh \beta^{k}_u)^{-1}d\wh \beta^{k}_u .
\]
When the cash account $B^0$ is used for funding and remuneration for adjustment processes, that is,
when $B^i=B^0$ for $i=1,2,\ldots,d$ and $\beta^k=B^0$ for $k=1,2,\ldots,n$, then the funding adjustment vanishes,
as was expected.

\subsection{Financial Interpretation of Trading Adjustments}     \label{sec2.4}

In this study, we will devote significant attention to terms appearing in the dynamics of $V^p(x,\varphi,A,\cX)$,
which correspond to the trading adjustment process $\cX$.

\begin{definition}
The stochastic process $\varpi=\sum_{k=1}^{n} \varpi^k $, where for $k=1,2, \ldots, n$,
\begin{equation}\label{eq:cashAdjustment}
 \varpi^k_t := \alpha^k_t X^k_t-\int_0^t X_u^k (\beta^{k}_u)^{-1}\,d\beta^{k}_u
\end{equation}
is called the \textit{cash adjustment}.
\end{definition}

In general, the financial interpretation of the cash adjustment term $\varpi^k$ is as follows: the term $\alpha^k_tX^k_t$ represents the part of the $k$th adjustment that the hedger can either use for his trading purposes when $\alpha^k_t X^k_t >0$
or has to put aside (for instance, pledge to his counterparty as a collateral or hold in a separate account as a regulatory capital)
when $\alpha^k_t X^k_t <0$.  Formally, the quantity $-X_t^k (\beta^{k}_t)^{-1}$ can be seen as the number of ``shares'' of the remuneration process $\beta^k$ that the hedger should hold at time $t$ in order to cover interest payments associated with the adjustment process $X^k$. Hence the integral $\int_0^t X_u^k (\beta^{k}_u)^{-1}\, d\beta^{k}_u$ represents the cumulative interest either paid or received by the hedger due to the presence of the $k$th trading adjustment.

Let us illustrate a few alternative interpretations of cash adjustments given by \eqref{eq:cashAdjustment}.
We hereafter write $\wh{X}^k=(\beta^k)^{-1} X^k $.
\begin{itemize}
\item Let us first assume that $\alpha^k_t=1$, for all $t$.  The term $X^k_t-\int_0^t \wh{X}^k_u\, d\beta^k_u$ indicates that the cash adjustment $\varpi^k$ is affected by both the current value $X^k_t$ of the adjustment process and by the cost of funding of this adjustment given by the integral $\int_0^t \wh{X}^k_u\, d\beta^k_u$. Such a situation occurs, for example, when $X^k$ represents the capital charge or the rehypotecated collateral. The integration by parts formula gives
   \begin{equation}\label{eq:T1}
      \varpi^k_t = X^k_t-\int_0^t \wh{X}^k_u\, d\beta^k _u=X^k_0+\int_0^t\, \beta^k_u\, d\wh{X}^k_u ,
   \end{equation}
     where the integral $\int_0^t \beta^k_u\,d\wh{X}^k_u$ has the following financial interpretation: $\wh{X}^k_u$ is the number of units of the funding account $\beta^k_u$ that are needed to fund the amount $X^k_u$ of the adjustment process. Hence $d\wh{X}^k_u$ is the infinitesimal change of this number and $\beta^k_u\,d\wh{X}^k_u$ is the cost of this change, which has to be absorbed by the change in the value of the trading strategy. Observe that the term $\beta^k_u\,d\wh{X}^k_u$ may be negative, meaning that a cash relieve situation is taking place.
\item In the special case when $\alpha^k_t=1$ and $\beta^k_t=1$ for all $t$, we obtain $\varpi^k_t=X^k_t$ for all $t$. We deal here with the cash adjustment $X^k$ on which there is no remuneration since manifestly $\int_0^t \wh{X}^k_u\,d\beta^k_u=0.$ This situation may arise, for example, if the bank does not use any external funding for financing this adjustment, but relies on its own cash reserves, which are assumed to be kept idle and neither yield interest nor require interest payouts.
\item Let us now assume that $\alpha^k_t=0$ for all $t$. Then the term $ \varpi^k_t=-  \int_0^t \wh{X}^k_u\, d\beta^k_u$ indicates that the cash value of the adjustment $X^k$ does not contribute to the portfolio value. Only the remuneration of the adjustment process $X^k$, which is given by the integral $\int_0^t \wh{X}^k_u\, d\beta^k_u$, contributes to the portfolio's value. This happens, for example, when the adjustment process represents the collateral posted by the counterparty and kept in the segregated account.
\end{itemize}

The above considerations lead to the following lemma, which gives a convenient representation for the cash adjustment process when
$\alpha^k$ is equal to either 1 or 0. In most practical situations, a general case can also be dealt with using Lemma \ref{lmnb} and a suitable redefinition of adjustment processes.

\begin{lemma} \label{lmnb}
Let the non-negative integers $n_1,n_2,n_3$ be such that $n_1+n_2+n_3=n$. If $\alpha^k=1$ for $k=1,2,\dots , n_1+n_2$,
 $\beta^k=1$ for $k=n_1, n_1+1, \dots , n_1+n_2$ and  $\alpha^k=0$ for $k=n_1+n_2, n_1+n_2+1, \dots , n_1+n_2+n_3$, then
 the cash adjustment process $\varpi$ satisfies, for all $t \in [0,T]$,
\begin{equation}\label{cashadj}
\varpi_t= \sum_{k=1}^{n_1}X_0^k+ \sum_{k=1}^{n_1}\int_0^t \beta^k_u\, d\wh{X}^k_u
+\sum_{k=n_1+1}^{n_1+n_2} X^k_t-\sum_{k=n_1+n_2+1}^{n_1+n_2+n_3}\int_0^t \wh{X}^k_u\,d\beta^k_u .
\end{equation}
\end{lemma}

\subsection{Wealth Process}     \label{sec2.5}

Let $(x, p,\phi,\cC)$ be an arbitrary self-financing trading strategy. Then the following natural question arises: what is the wealth of a hedger at time $t$, say $V_t(x, p,\phi,\cC)$? It is clear that if the hedger's initial endowment equals $x$, then his initial wealth equals $x+p$ when he sells a contract $\cC $ at the price $p$ at time 0. By contrast, the initial value of the hedger's portfolio, that is, the total amount of cash he invests at time 0 in his portfolio of traded assets,
is given by \eqref{xsxs} meaning that the trading adjustments at time 0 need to be accounted for in the initial portfolio's value. However, according to the financial interpretation of trading adjustments, they have no bearing on the hedger's initial wealth and thus the relationship between the hedger's initial wealth and the initial portfolio's value reads
 \[
 V_0 (x,p,\phi,\cC)=V^p_0 (x,p,\phi,\cC)-\sum_{k=1}^n \alpha^k_0 X^k_0 .
 \]
Analogous arguments can be used at any time $t\in [0,T]$, since the hedger's wealth at time $t$ should represent the value of his portfolio of traded assets net of the value of all trading adjustments (see \eqref{wealth2}). Furthermore, one needs to focus on the actual ownership (as opposed to the legal ownership) of each of the adjustment processes $X^1, \dots , X^n$, of course, provided that they do not vanish at time~$t$. Although this general rule is cumbersome to formalize, it will not present any difficulties when applied to a particular contract at hand.

For instance, in the case of the rehypothecated cash collateral (see Section \ref{sec2.7.1}), the hedger's wealth at time $t$ should be computed by subtracting the collateral amount $C_t$ from the portfolio's value. This is consistent with the actual ownership of the cash amount delivered by either the hedger or the counterparty at time $t$. For example, if $C^+_t>0$ then the {\it legal owner} of the amount $C^+_t$ at time $t$ could be either the holder or the counterparty (depending on the legal covenants of the collateral agreement) but the hedger, as a collateral taker, is allowed to use the collateral amount for his trading purposes. If there is no default before $T$, the collateral taker returns the collateral amount to the collateral provider. Hence the amount $C^+_t$ should be accounted for when dealing with the hedger's portfolio,  but should be excluded from his wealth. In general, we have the following definition of the wealth process.

\begin{definition} \label{defwealth}
The {\it wealth process} of a self-financing trading strategy $(x_t, p_t ,\phi^t,\cC^t)$ is
defined, for every $u \in [t,T]$, by
\begin{equation} \label{wealth2x}
V_u(x_t,p_t,\phi^t,\cC^t):= V^p_u(x_t,p_t,\phi^t,\cC^t)-\sum_{k=1}^n \alpha^k_u X^k_u
\end{equation}
or, more explicitly,
\begin{equation} \label{wealth2}
V_u(x_t,p_t,\phi^t,\cC^t)=\sum_{i=1}^{d}\xi^i_u S^i_u+\sum_{j=0}^{d}\left( \psi^{j,l}_u B^{j,l}_u+\psi_u^{j,b} B^{j,b}_u\right)-\sum_{k=1}^n \alpha^k_u X^k_u .
\end{equation}
\end{definition}

Let us observe that there is a lot of flexibility in the choice of the adjustment processes $X^k$ and corresponding processes $\alpha^k$. However, we will always assume that these processes are specified such that the above arguments of interpreting the actual ownership of the capital and thus also of the wealth process $V(x,p,\phi, A,\mathcal{X})$ hold true.

As an immediate consequence of Definitions~\ref{ts2} and \ref{defwealth}, it follows that the wealth process $V$
of any self-financing trading strategy $(x_t, p_t,\phi^t,\cC^t)$ admits the dynamics, for  $u \in [t,T]$,
\begin{equation}\label{wealth3-1}
\begin{aligned}
V_u(x_t,p_t,\phi^t,\cC^t)
= x_t+p_t   &+\sum_{i=1}^{d}  \int_t^u \xi^i_v\,d(S^i_v+D^i_v )+ \sum_{j=0}^{d}  \int_t^u \left( \psi^{j,l}_v\,dB^{j,l}_v
+\psi_v^{j,b}\, dB^{j,b}_v \right) \\  &-\sum_{k=1}^{n}  \int_t^u X_v^k (\beta^{k}_v)^{-1}\,d\beta^{k}_v+ A^t_u .
\end{aligned}
\end{equation}

One could argue that  it would be possible to take equations \eqref{wealth2} and \eqref{wealth3-1} as the definition of a self-financing trading strategy and subsequently deduce that equality \eqref{eq:selfFin1} holds for the portfolio's value $V^p(x,p,\phi,\cC)$, which is then given by \eqref{wealth2x}. We contend this alternative approach would not be optimal, since conditions in Definition \ref{ts2} are obtained through a straightforward analysis of the trading mechanism and physical cash flows, whereas the financial justification of equations \eqref{wealth2}--\eqref{wealth3-1} is less appealing.

Clearly, the wealth processes of a self-financing trading strategy is characterized in terms of two equations \eqref{wealth2} and \eqref{wealth3-1}. Observe that, using \eqref{wealth2}, it is possible to eliminate one of the processes $\psi^{j,l}$ or $\psi^{j,b}$ from \eqref{wealth3-1} and thus to characterize  the wealth process in terms of a single equation. One obtains in that way a (typically nonlinear) BSDE, which can be used to formulate various valuation problems for a given contract.

\subsection{Trading in Risky Assets}           \label{sec2.6}

Note that we do not postulate that the processes $S^i,\, i=1,2, \dots, d$ are positive, unless it is explicitly stated that the process $S^i$ models the price of a stock. Hence by the {\it long cash position} (resp.  {\it short cash position}), we mean the situation when $\xi^i_t S^i_t \leq 0$ (resp.   $\xi^i_t S^i_t \geq 0$), where $\xi^i_t$ is the number of hedger's positions in the risky asset $S^i$ at time $t$.

\subsubsection{Cash Market Trading} \label{sec2.6.1}

For simplicity of presentation, we assume that $S^i_t \geq 0$ for all $t \in [0,T]$.
Assume first that the purchase of $\xi^i_t > 0$ shares of the $i$th risky asset is funded using cash.
Then, we set $\psi^{i,b}_t=0$ for all $t \in [0,T]$ and thus the process $B^{i,b} $ becomes irrelevant. Let us now consider the case when $\xi^i_t <0$. If we assume that the proceeds from short selling of the risky asset $S^i$ can be used by the hedger (this is typically not true in practice), we also set $\psi^{i,l}_t=0$ for all $t \in [0,T]$, and thus the process $B^{i,l}$ becomes irrelevant as well. Hence, under these \textit{stylized}  cash trading
conventions, there is no need to introduce the funding accounts $B^{i,l}$ and $B^{i,b}$ for the $i$th risky asset.
Since dividends $D^i$ are passed over to the lender of the asset, they do
not appear in the term representing the gains/losses from the short position in the risky asset.
In the simplest case of no market frictions and trading adjustments,
and with the single risky asset $S^1$, under the present short selling convention,  \eqref{eq:selfFin1} becomes
$$
V^p_t (x,p_0,\phi,\cC)= x+p+ \int_0^t \xi^1_u \, (dS^1_u+dD^1_u )
+  \int_0^t \left( \psi^{0,l}_u\,dB^{0,l}_u+\psi_u^{0,b}\, dB^{0,b}_u \right)+ A_t .
$$
More practical short selling conventions for risky assets are discussed in the foregoing subsections.

\subsubsection{Short Selling of Risky Assets}   \label{sec2.6.2}

Let us now consider the classical way of short selling of a risky asset borrowed from
a broker. In that case, the hedger does not receive the proceeds from the sale of the borrowed shares of
a risky asset, which are held instead by the broker as the cash collateral. The hedger may also be requested to
post additional cash collateral to the broker and, in some cases, he may be paid interest on his margin account with the broker.\footnote{The interested reader may consult the web pages \url{http://www.investopedia.com/terms/s/shortsale.asp} and
\url{https://www.sec.gov/investor/pubs/regsho.htm} for more details on the mechanics of short-sales.}
To represent these trading arrangements for the $i$th risky asset, we set $\psi^{i,l}_t=0$,  $\alpha^i_t=\alpha_t^{i+d}=0$ and
$$
X^{i}_t=-(1+\delta^i_t)(\xi^i_t)^- S^i_t, \quad X^{i+d}=\delta^i_t (\xi^i_t)^- S^i_t
$$
where $\beta^{i}_t$ specifies the interest (if any) on the hedger's margin account with the broker,
$\delta^i_t \geq 0$ represents an additional cash collateral, and $\beta^{i+d}$ specifies the interest rate
paid by the hedger for financing the additional collateral. Let us assume, for instance, that there is only one risky
asset, $S^1$, which is either sold short or purchased using cash as in Section~\ref{sec2.6.1}. Then we obtain the
following expression for the portfolio value
\begin{equation} \label{xwea1}
V^p_t(x,p,\phi,\cC)=(\xi^1_t)^+ S^1_t+\psi^{0,l}_t B^{0,l}_t+\psi_t^{0,b} B^{0,b}_t
\end{equation}
whereas equation \eqref{eq:selfFin1} becomes
\begin{equation}\label{weass2}
\begin{aligned}
V^p_t (x,p,\phi,\cC)= x &+ p+ \int_0^t \xi^1_u \, (dS^1_u +dD^1_u )
+\int_0^t \left( \psi^{0,l}_u\,dB^{0,l}_u+\psi_u^{0,b}\, dB^{0,b}_u \right)+ A_t \\
&+\int_0^t (\beta^1_u)^{-1}(1+  \delta^{1}_u)(\xi^1_u)^- S^1_u\,d\beta^1_u
-\int_0^t (\beta^2_u)^{-1}\delta^{1}_u(\xi^1_u)^- S^1_u\,d\beta^2_u.
\end{aligned}
\end{equation}
In particular, if there is no specific interest rate for remuneration of an additional collateral, then
we set $X^2=0$ and thus the last term in  \eqref{weass2} should be omitted.
It is worth noting that \eqref{xwea1} can be seen as a special case of the following extended version of \eqref{wea1}
\begin{equation*}
V^p_u(x_t,p_t,\phi^t,\cC^t):=\sum_{i=1}^{d} h_i(\xi^i_u )S^i_u+\sum_{j=0}^{d}\left( \psi^{j,l}_u B^{j,l}_u+\psi_u^{j,b} B^{j,b}_u\right)
\end{equation*}
with $d=1$ and $h_1(x)=x^+$ for $x \in \bR$.

\subsubsection{Repo Market Trading}    \label{sec2.6.3}

Let us first consider the \textit{cash-driven repo transaction}, the situation when shares of the $i$th risky asset owned by the hedger are used as collateral to raise cash.\footnote{We refer to \url{https://www.newyorkfed.org/medialibrary/media/research/staff_reports/sr529.pdf}
for a detailed description of mechanics of repo trading.} To represent this transaction, we set
\begin{equation} \label{repo1a}
\psi^{i,b}_t= -(B^{i,b}_t)^{-1}(1-h^{i,b})(\xi^i_t)^+ S^i_t,
\end{equation}
where $B^{i,b}$ specifies the interest paid to the lender by the hedger who borrows cash and pledges the risky asset $S^i$ as collateral, and the constant $h^{i,b}$ represents the {\it haircut} for the $i$th asset pledged.

A synthetic short-selling of the risky asset $S^i$ using the repo market is obtained through the \textit{security-driven
repo transaction,} that is, when shares of the risky asset are posted as collateral by the borrower of cash
and they are immediately sold by the hedger who lends the cash. Formally, this situation corresponds to the equality
\begin{equation}  \label{repo1b}
\psi^{i,l}_t= (B^{i,l}_t)^{-1}(1-h^{i,l})(\xi^i_t)^- S^i_t,
\end{equation}
where $B^{i,l}$ specifies the interest amount paid to the hedger by the borrower of the cash amount $(1-h^{i,l})(\xi^i_t)^- S^i_t$ and $h^{i,l}$ is the corresponding haircut.

If only one risky asset is traded and transactions are exclusively in repo market, then we obtain
 \begin{equation}\label{weass2a}
\begin{aligned}
&V^p_t (x,p,\phi,\cC)= x+p+ \int_0^t \xi^1_u \, (dS^1_u+dD^1_u )+\int_0^t \left( \psi^{0,l}_u\,dB^{0,l}_u+\psi_u^{0,b}\, dB^{0,b}_u \right) \\
&+\int_0^t \left( (B^{1,l}_u)^{-1} (1-h^{1,l})(\xi^1_u)^- S^1_u\,dB^{1,l}_u-(B^{1,b}_u)^{-1}(1- h^{1,b})(\xi^1_u)^+ S^1_u\,dB^{1,b}_u \right)+A_t .
\end{aligned}
\end{equation}

In practice, it is reasonable to assume that the long and short repo rates for a given risky asset are identical,
that is, $B^i=B^{i,l}=B^{i,b}$. In that case, we may and do set $\psi^i_t =(1-h^i)(B^i_t)^{-1} \xi^i_t S^i_t $, so that equations \eqref{repo1a} and  \eqref{repo1b} reduce to just one equation
\begin{equation} \label{repo3}
(1-h^i) \xi^i_t S^i_t+\psi^i_t B^i_t=0.
\end{equation}
According to this interpretation of $B^i$, equality \eqref{repo3} means that trading in the $i$th risky asset is done using the (symmetric) repo market and $\xi^i_t$ shares of a risky asset are pledged as collateral at time $t$, meaning that the collateralization rate equals 1.  Under \eqref{repo3}, equation \eqref{weass2a} reduces to
\begin{equation}\label{weass2b}
\begin{aligned}
V^p_t(x,p,\phi,\cC)=x&+p+\int_0^t\xi^1_u\,(dS^1_u+dD^1_u)+\int_0^t \left(\psi^{0,l}_u\,dB^{0,l}_u+\psi_u^{0,b}\,dB^{0,b}_u\right)
\\&-\int_0^t(B^1_u)^{-1}(1-h^1)\xi^1_u S^1_u\,dB^1_u+A_t.
\end{aligned}
\end{equation}

\subsection{Collateralization}      \label{sec2.7}

We consider the situation when the hedger and the counterparty enter a contract and either receive or pledge \textit{collateral} with value denoted by $C$, which is assumed to be a semimartingale. Generally speaking, the process $C$ represents the value of the \textit{margin account}. We let
\begin{equation} \label{collss}
C_t=X^1_t+X^2_t
\end{equation}
where $X^1_t:=C_t^+=C_t \I_{\{ C_t \geq 0\}}$, and $X^2_t := -C_t^-= C_t \I_{\{ C_t < 0\}}$.
By convention, the amount $C^+_t$ is the cash value of collateral received at time $t$ by the hedger from the counterparty, whereas $C^-_t$ represents the cash value of the collateral pledged by him and thus transferred to his counterparty. For simplicity of presentation and consistently with the prevailing market practice, it is postulated throughout that only cash collateral may be delivered or received (for other collateral conventions, see, e.g., Bielecki and Rutkowski \cite{BR2015}).
According to ISDA Margin Survey 2014, about 75\% of non-cleared OTC collateral agreements are settled in cash and about 15\% in government securities. We also make the following natural assumption regarding the value of the margin account at the contract's maturity date.

\begin{assumption} \label{assumargin}
{\rm The $\bG$-adapted collateral amount process $C$ satisfies $C_T=0$.}
\end{assumption}

Typically this means that the collateral process $C$ will have a jump at time $T$ from $C_{T-}$ to 0.
The postulated equality $C_T=0$ is simply a convenient way of ensuring that any collateral amount posted is returned in full to the pledger when the  contract matures, provided that default events have not occurred prior to or at maturity date $T$.
As soon as the default events are also modeled, we will need to specify closeout payoffs (see Section~\ref{sec2.8.1}).

Let us first make some comments from the hedger's perspective regarding the crucial features of the margin account. The financial practice may require to hold the collateral amounts in {\it segregated} margin accounts, so that the hedger, when he is a collateral taker, cannot make use of the collateral amount for trading. Another collateral convention mostly encountered  in practice is {\it rehypothecation}  (around 90\% of cash collateral of OTC contracts are rehypothecated), which refers to the situation where the hedger may use the collateral pledged by his counterparties as collateral for his contracts
with other counterparties. Obviously, if the hedger is a collateral provider, then a particular convention regarding segregation or rehypothecation is immaterial for the dynamics of the value process of his portfolio. We refer the reader to Bielecki and Rutkowski \cite{BR2015} and Cr\'epey et al. \cite{CBB2014} for a detailed analysis of various conventions on collateral agreements. Here we will examine some basic aspects of \textit{collateralization} (sometimes also called \textit{margining}) in our context.

In general, the cash adjustments due to collateralization are
\begin{equation}\label{eq:collateral1}
\varpi^C_t :=\alpha_t^1C^+_t-\alpha_t^2C^-_t-\int_0^t(\beta^1_u)^{-1}C^+_u\,d\beta_u^1+\int_0^t(\beta^2_u)^{-1}C^-_u\,d\beta_u^2 ,
\end{equation}
where the remuneration processes $\beta^1$ and $\beta^2$ determine the interest rates paid or received by the hedger  on  collateral amounts $C^+$ and $C^-$, respectively. The auxiliary processes $\alpha^1$ and $\alpha^2$ introduced
in \eqref{eq:collateral1} are used to cover alternative conventions regarding rehypothecation and segregation of margin accounts.
Note that we always set $\alpha^2_t=1$  for all $t \in [0,T]$ when considering the portfolio of the hedger, since a particular convention regarding rehypothecation or segregation is manifestly irrelevant for the pledger of collateral.

\subsubsection{Rehypothecated Collateral}       \label{sec2.7.1}

As it is customary in the existing literature, we assume that rehypothecation of cash collateral means that it can be used by the hedger for his trading purposes without any restrictions. To cover this stylized version of a {\it rehypothecated collateral} for the hedger, it suffices to set $\alpha^1_t=\alpha^2_t =1$ for all $t \in [0,T]$, so that for the hedger we obtain $\alpha_t^1X^1_t+\alpha_t^2X^2_t=C_t$. Consequently, the cash adjustment corresponding to the margin account equals
\begin{equation}\label{cashadjrehyp}
\varpi_t=\varpi^1_t+\varpi^2_t=\sum_{k=1}^{2}\Big( X^k_0+\int_0^t \beta^{k}_u\, d\wh{X}^k_u\Big).
\end{equation}

\subsubsection{Segregated Collateral}             \label{sec2.7.2}

Under segregation, the collateral received by the hedger is kept by the third party, so that it cannot be used by the hedger for his trading activities. In that case, we set $\alpha^1_t=0$ and $\alpha^2_t =1$ for all $t \in [0,T]$ and thus  $\alpha_t^1X^1_t+\alpha_t^2X^2_t=-C^-_t$. Hence the corresponding cash adjustment term $\varpi$ equals
\begin{equation}\label{cashadjrehyp1}
\varpi_t= \varpi^1_t+\varpi^2_t=X^2_0-\int_0^t \wh{X}^1_u\, d\beta^1_u+\int_0^t \beta^2_u\, d\wh{X}^2_u.
\end{equation}

\subsubsection{Initial and Variation Margins}     \label{sec2.7.3}

In market practice, the total collateral amount is usually represented by two components, which are termed the {\it initial margin} (also known as the {\it independent amount}) and the {\it variation margin}. In the context of self-financing trading strategies, this can be easily dealt with by introducing two (or more) collateral processes for a given contract $A$. It is worth mentioning that each of the collateral processes specified in the clauses of a contract is usually subject to a different convention regarding segregation and/or remuneration.

\subsection{Counterparty Credit Risk}      \label{sec2.8}

The {\it counterparty credit risk} in a financial contract arises from the possibility that at least one of the parties in the contract may default prior to or at the contract's maturity, which may result in failure of this party to fulfil all their contractual obligations leading to financial loss suffered by either one of the two parties in the contract. We will model defaultability of the two parties to the contract in terms of their default times. We denote by $\tau^h$ and $\tau^c$ the default times of the hedger and his counterparty, respectively. We require that $\tau^h$ and $\tau^c$ are non-negative random variables defined on $(\Omega, \cG, \bG , \bP)$. If $\tau^h>T$ holds a.s. (resp. $\tau^c>T$, a.s.) then the hedger (resp. the counterparty) is considered to be default-free in regard to the contract under study. Hence the counterparty risk is a relevant aspect for the contract maturing at $T$ provided that $\bP(\tau\leq T)>0$ where $\tau:= \tau^h \wedge \tau^c$ is the moment of the first default.

From now on, we postulate that the process $A$ models all promised (or \textit{nominal}) cash flows of the contract, as seen from the perspective of the trading desk without accounting for the possibility of defaults of trading parties. In other words, $A$ represents cash flows that would be realized in case none of the two parties has defaulted prior to or at the contract's maturity. We will sometimes refer to $A$ as to the \textit{counterparty risk-free cash flows} and we will call the contract with cash flows $A$ the \textit{counterparty risk-free contract}. The key concept in the context of counterparty risk is the \textit{counterparty risky contract}, which will be examined in the foregoing subsection.

\subsubsection{Closeout Payoff}          \label{sec2.8.1}

Recall that $\tau $ denotes the moment of the first default. On the event $\{ \tau<\infty \}$, we define the random variable $\Upsilon$ as
\begin{align} \label{e:chisup0}
\Upsilon=Q_{\tau}+ \Delta A_\tau - C_\tau ,
\end{align}
where $Q$ is the Credit Support Annex (CSA) closeout valuation process of the contract $A$, $\Delta A_\tau =A_\tau - A_{\tau-}$ is the jump of $A$ at $\tau$ corresponding to a (possibly null) promised bullet dividend at $\tau$, and $C_\tau$ is the value of the collateral process $C$ at time $\tau$. In the financial interpretation, $\Upsilon^+$ is the amount the counterparty owes to the hedger at time $\tau $, whereas $\Upsilon^-$ is the amount the hedger owes to the counterparty at time $\tau $. It accounts for the legal value $Q_\tau$ of the contract, plus the bullet dividend $\Delta A_\tau$ to be received/paid at time $\tau$, less the collateral amount $C_\tau$ since it is already held by either the hedger (if $C_\tau>0$) or the counterparty (if $C_\tau<0$). We refer the reader to Section~3.1.3 in Cr\'epey et al. \cite{CBB2014} for the detailed discussion of the specification of $\Upsilon$.

One of the key financial aspects of the counterparty credit risk is the \textit{closeout payoff}, which occurs if at least one of the parties defaults either before or at the maturity of the contract. It represents the cash flow exchanged between the two parties at the first-party-default time. The following definition of the closeout payoff, as usual given from the perspective of the hedger, is taken from Cr\'epey et al. \cite{CBB2014}. The random variables $R_c$ and $R_h$ taking values in $[0,1]$ represent the recovery rates of the counterparty and the hedger, respectively.

\begin{definition} \label{close}
The  {\it CSA closeout payoff} $\mathfrak{K}$ is defined as
\begin{equation} \label{closeout}
\mathfrak{K} := C_\tau+\I_{\{\tau^c< \tau^h\}}(R_c\Upsilon^+-\Upsilon^-) +
 \I_{\{\tau^h< \tau^c \}}(\Upsilon^+ - R_h\Upsilon^-)+\I_{\{\tau^h=\tau^c\}}( R_c \Upsilon^+ - R_h\Upsilon^-) .
\end{equation}
The {\it counterparty risky cumulative cash flows} process $A^\sharp$ is given by
\begin{equation} \label{hatA}
A^\sharp_t=\I_{\{t<\tau\}} A_t+ \I_{\{t\geq \tau\}}(A_{\tau-}+\mathfrak{K}),\quad t\in [0,T].
\end{equation}
\end{definition}

Let us make some comments on the form of the closeout payoff $\mathfrak{K}$. First, the term $C_{\tau }$ is due to the fact that the legal title to the collateral amount comes into force only at the time of the first default.  The three terms appearing after $C_{\tau }$ in \eqref{closeout} correspond to the CSA convention that the cash flow at the first default from the perspective of the hedger should be equal to $Q_{\tau}+ \Delta A_\tau $. Let us consider, for instance, the event $\{ \tau_c < \tau_h \}$. If $\Upsilon^+ >0$, then we obtain
$$
\mathfrak{K}=C_\tau+R_c(Q_{\tau}+ \Delta A_\tau - C_\tau ) \leq Q_{\tau}+ \Delta A_\tau ,
$$
where the equality holds whenever $R_c =1$. If $\Upsilon^- >0$, then we get
$$
\mathfrak{K}=C_\tau - (-Q_{\tau}- \Delta A_\tau+C_\tau )= Q_{\tau}+ \Delta A_\tau .
$$
Finally, if $\Upsilon =0$, then $\mathfrak{K}=C_\tau=Q_{\tau}+ \Delta A_\tau $.
Similar analysis can be done on the remaining two events in \eqref{closeout}.

\begin{remark} \label{remccr}
{\rm Of course, there is no counterparty credit risk present under the assumption that $\bP(\tau >T)=1$. Let
us consider the case where  $\bP(\tau >T) < 1$. We denote by $P^e_t$ the {\it counterparty risk-free}
ex-dividend price of the contract at time $t$. If we set $R_c=R_h=1$, then we obtain
\[
A^\sharp_{\tau }=A_{\tau }+Q_\tau .
\]
Hence the counterparty credit risk is still present, despite the postulate of the full recovery,
unless the legal value $Q_{\tau }$ perfectly matches the counterparty risk-free ex-dividend price $P^e_{\tau }$. Obviously,
the counterparty credit risk vanishes when $R_c=R_h=1$ and $Q_\tau= P^e_\tau $,
since in that case the so-called {\it exposure at default} (see Section~3.2.3 in Cr\'epey et al. \cite{CBB2014}) is null.}
\end{remark}

\subsubsection{Counterparty Credit Risk Decomposition}  \label{sec2.8.2}

To effectively deal with the closeout payoff in our general framework, we now define the counterparty credit risk (CCR) cash flows,
which are sometimes called CCR exposures. Note that the events $\{ \tau=\tau^h \}=\{ \tau^h \leq \tau^c \}$
and  $\{ \tau=\tau^c \}=\{ \tau^c \leq \tau^h \}$ may overlap.

\begin{definition} \label{expo}
By the {\it CCR processes,} we mean the processes  $CL, CG$ and $RP$ where the {\it credit loss} $CL$ equals
\[
CL_t=- \I_{\{t \geq \tau \}}\I_{\{\tau=\tau^c\}}(1-R_c)\Upsilon^+ ,
\]
the {\it credit gain} $CG$ equals
\[
CG_t=  \I_{\{t \geq \tau \}}\I_{\{\tau=\tau^h\}}(1-R_h)\Upsilon^- ,
\]
and the {\it replacement process} is given by
\[
CR_t=  \I_{\{t\geq \tau\}}(A_{\tau } - A_t+Q_\tau ).
\]
The {\it CCR cash flow} is given by $A^{\rm CCR}=CL+CG+CR$.
\end{definition}

It is worth noting that the CCR cash flows depend on the processes $A, C$ and $Q$.
The next proposition shows that we may interpret the counterparty risky contract as the basic contract $A$, which is complemented
by the collateral adjustment process $\cX=(X^1,X^2)= (C^+,-C^-)$ and the CCR cash flow $A^{\rm CCR}$. In view of this result,
the counterparty risky contract $(A^\sharp ,\cX )$ admits the following formal decompositions  $(A^\sharp ,\cX )
= (A,\cX )+(A^{\rm CCR},0)$ and  $(A^\sharp ,\cX )=(A,0)+(A^{\rm CCR},\cX )$.

\begin{proposition}\label{pro1}
The equality $A^\sharp_t=A_t+A^{\rm CCR}_t$ holds for all $t \in [0,T]$.
\end{proposition}

\proof
We first note that
\begin{align*}
\mathfrak{K} &=C_\tau+\I_{\{\tau^c< \tau^h\}}(R_c\Upsilon^+-\Upsilon^-) +
 \I_{\{\tau^h< \tau^c \}}(\Upsilon^+ - R_h\Upsilon^-)+\I_{\{\tau^h=\tau^c\}}( R_c \Upsilon^+ - R_h\Upsilon^-) \\
&= C_\tau - \I_{\{\tau^c \leq \tau^h\}}(1-R_c)\Upsilon^+ +
 \I_{\{\tau^h \leq \tau^c \}}(1-R_h)\Upsilon^-+\Upsilon \\
&= Q_\tau+\Delta A_\tau - \I_{\{\tau^c \leq \tau^h\}}(1-R_c)\Upsilon^+ +
 \I_{\{\tau^h \leq \tau^c \}}(1-R_h)\Upsilon^- ,
\end{align*}
where we used \eqref{e:chisup0} in the last equality. Therefore, from \eqref{hatA} we obtain
\begin{align*}
A^\sharp_t &= \I_{\{t<\tau\}}A_t+ \I_{\{t\geq \tau\}}(A_{\tau-}+\mathfrak{K})= \I_{\{t<\tau\}}A_t+ \I_{\{t\geq \tau\}}(A_{\tau} -  \Delta A_\tau +\mathfrak{K}) \\
&= A_{t \wedge \tau }+ \I_{\{t\geq \tau\}}(\mathfrak{K} -  \Delta A_\tau)=A_t +( A_{t \wedge \tau } - A_t)+ \I_{\{t\geq \tau\}}(\mathfrak{K} -  \Delta A_\tau) \\
&= A_t + \I_{\{t\geq \tau\}} \big(A_{\tau } - A_t+Q_\tau - \I_{\{\tau^c \leq \tau^h\}}(1-R_c)\Upsilon^+ +
 \I_{\{\tau^h \leq \tau^c \}}(1-R_h)\Upsilon^- \big),
\end{align*}
which is the desired equality in view of Definition \ref{expo}.
\endproof

Proposition \ref{pro1} shows that cash flows of the counterparty risky contract can be formally decomposed into the counterparty risk-free component $(A^1,\cX^1)=(A,\cX )$ and the CCR component $(A^2 ,\cX^2)=(A^{\rm CCR},0)$. This additive decomposition of the contract's cash flows may be employed in pricing of a counterparty risky contract. For instance, one could attempt to compute the price of the contract $(A^\sharp ,\cX)$ using the following tentative decomposition
\begin{equation*}
\textrm{price}\,(A^\sharp,\cX )\,=\, \textrm{price}\,(A,\cX)+\textrm{price}\,(A^{\rm CCR},0) \,=\,\textrm{counterparty risk-free price }+\, \textrm{CCR price}.
\end{equation*}
It is unlikely that this procedure would result in an overall arbitrage-free valuation of the counterparty risky contract in a
nonlinear framework since, as we argue in Section \ref{sec6}, the additivity of ex-dividend prices obtained by solving nonlinear BSDEs fails to hold, in general.

%

\subsection{Local and Global Valuation Problems}    \label{sec2n1}

Market adjustments, represented in our framework by the process $\cX$,  may in fact depend both on the cash flow process $A$ and the trading strategy $\varphi$. By the same token, the trading strategy $\varphi$ will typically depend on the trading adjustments. So, a feedback effect between $\varphi$ and $\cX$ is potentially present in our trading universe and, of course, this feature should be properly accounted for in valuation and hedging. Furthermore, it is important to distinguish between the case where the above-mentioned dependence is only on the current composition of the hedging portfolio and/or the current level of the wealth process and the case, where this dependence extends to the history of a hedging strategy. If a contract $(A,\cX)$, the cash and funding accounts, and the prices of risky assets do not depend on the strict history (i.e., the history not including the current values of processes of interest) of a hedger's trading strategy $\phi $ and the wealth process $V(\phi)$, then we say that a valuation problem is {\it local}; otherwise, it is referred to as a {\it global} valuation problem. In view of \eqref{wealth3-1}, the distinction between local and global valuation problems can be formalized as follows.

\begin{definition}   \label{locglob}
 A valuation problem is said to be {\it local} if $X^k_t=v^k(t, V_t(\phi) ,\phi_t)$ and $d\beta_t^{k}=w^{k}(t,V_t(\phi),\phi_t)\,dt$ for some $\mathbb{G}$-progressively measurable mappings $v^k,w^k:\Omega \times [0,T] \times \bR^{3(d+1)} \to \bR $ for every $k=1,2,\dots ,n$. A valuation problem is said to be {\it global} if  $X^k_t=\bar{v}^k(t,V_{\cdot }(\phi),\phi_{\cdot})$ and $d\beta_t^{k}= \bar{w}^k(t,V_{\cdot } (\phi),\phi_{\cdot})\,dt$ for some $\mathbb{G}$-non-anticipative functionals $\bar{v}^k,\bar{w}^k :
\Omega \times [0,T] \times \cD ([0,T], \bR^{3(d+1)}) \to \bR $ for every $k=1,2,\dots ,n$ where $ \cD ([0,T], \bR^{3(d+1)})$ is the space of $\bR^{3(d+1)}$-valued, $\mathbb{G}$-adapted, c\`adl\`ag processes on $[0,T]$.
\end{definition}

As one might guess, solutions to the two valuation problems will always coincide at time 0 but, in general, they may have very different properties at any date $t \in (0,T)$. In particular, they will typically correspond to different classes of BSDEs: local problems correspond to classical BSDEs, whereas global ones can be dealt with through generalized BSDEs, which were introduced in the recent work by Cheridito and Nam \cite{CN2015} (see also Zheng and Zong \cite{ZZ2017}). It is important to stress that the distinction between the local and global problems is not related to the concept of path-independent contingent claims or the Markov property of the underlying model for primary risky assets. It is only due to the above-mentioned (either local or global) feedback effect between the hedger's trading decisions and the market conditions inclusive of particular adjustments for the contract at hand.

\begin{example}  \label{exa2.9}
As a stylized example of a global valuation problem, let us consider a contract, which lasts for two months (for concreteness, assume that it is a simple combination of the put and the call on the stock $S^1$ with maturities equal to one month and two months, respectively). The borrowing rate for the hedger is set to be 5\% per annum, rising to 6\% after one month if the hedger borrows any cash during the first month and it will stay at 5\% if he does not. Similarly, the lending rate initially equals 3\% per annum and drops to 2\% if the hedger borrows any cash during the first month. It is intuitively clear that the valuation problem is here global, since its solution on $[t,T]$ will depend on the strict history of trading. In contrast, if the trading model has possibly different, but fixed, borrowing and lending rates, then the valuation problem for any contract will be local, in the sense introduced above, of course, unless some other trading adjustments will depend on the strict history of trading.

For instance, if the only adjustment is the variable margin account determined by the hedger's valuation and with a constant remuneration rate, then the hedger's valuation problem is local. Note that the valuation problem described above can be inherently global even when the stock price is governed under the real-world probability measure by Markovian dynamics and the contract under study is a standard call or put option (or any other path-independent contingent claim).

More general instances of local and global valuation problems are presented in Section \ref{sec6} where we examine a BSDE approach to the nonlinear markets. Let us mention that most valuation problems examined in the existing literature are local and thus they can be solved using existing results for classical BSDEs. In contrast, global valuation problems are much harder to analyze, since they require to use novel classes of BSDEs (see \cite{CN2015,ZZ2017} and the references therein).
\end{example}

\section{No-Arbitrage Properties of Nonlinear Markets}    \label{sec3}

The analysis of the self-financing property of a trading strategy should be complemented by the study
of some kind of the no-arbitrage property for the adopted market model. Due to the nonlinearity of a market model with differential funding rates, the question how to properly define the no-arbitrage property is already a nontrivial matter, even when no additional portfolio constraints or trading adjustments are taken into account. Nevertheless, we will argue that it can be effectively dealt with using some reasonably general definition of an {\it arbitrage opportunity} associated with trading. Let us stress that we only examine here a nonlinear extension of the classical concept of an arbitrage opportunity and hence the simplest definition of no-arbitrage, sometimes abbreviated as NA (see, for instance, part (iv) in Definition 2.2 in Fontana \cite{F2015}), as opposed to much more sophisticated concepts, such as: NFLVR (no free lunch with vanishing risk), NUPBR (no unbounded profit with bounded risk, which is also known as the no-arbitrage of the first kind, that is, NA1) or NIP (no increasing profit). The introduction of more sophisticated no-arbitrage conditions is motivated by the desire to establish a suitable version of the fundamental theorem of asset pricing (FTAP), which shows the equivalence between a particular form of no-arbitrage and the existence of some kind of a ``martingale measure'' for the discounted prices of primary assets. Due to the complexity of a general nonlinear market model, it is unlikely that the martingale technique underpinning the FTAP in the linear setup will also prove useful when working within the general nonlinear framework (see, however, Pulido \cite{P2014} who established the FTAP for a very special, and hence tractable, case of a nonlinear market with short sales prohibitions). In this paper, we only propose alternative definitions of no-arbitrage in a nonlinear framework and we give sufficient conditions for the no-arbitrage property of a general nonlinear market model.

\subsection{No-arbitrage Pricing Principles}   \label{sec3.1}

Let us first describe very succinctly the classical valuation paradigm for financial derivatives. In essence,
a general approach to the arbitrage-free pricing hinges, at least implicitly, on the following arguments:
\vskip 2 pt
\noindent {\bf Step (L.1).}  One first checks whether a market model with predetermined trading rules and primary
traded assets is arbitrage-free, where the definition of an arbitrage opportunity is a mathematical formalization of the real-world concept of a risk-free profitable trading opportunity. In fact, depending on the framework at hand, several alternative definitions of ``no-arbitrage'' were studied (for an overview, see Fontana \cite{F2015}).
\vskip 2 pt
\noindent {\bf Step (L.2).} Given a financial derivative for which the price is yet unspecified, one proposes a price
(not necessarily unique) and checks whether the extended model (that is, the model where
the financial derivative is postulated to be an additional traded asset) preserves the no-arbitrage property
in the sense made precise in Step (L.1).

The valuation procedure outlined above can be referred to as the {\it arbitrage-free pricing paradigm}.
In any linear market model (see the comments after Definition \ref{def:UnderMarket}), one can show
that the unique price given by replication (or the range of no-arbitrage prices obtained using the concept of superhedging
strategies in the case of an incomplete market) is consistent with the arbitrage-free pricing paradigm (L.1)--(L.2),
although to establish this property in a continuous-time framework, one needs also to introduce the concept of
{\it admissibility} of a trading strategy. In particular, the strict comparison property of linear BSDEs can be employed
to show that replication (or superhedging) will indeed yield prices for derivatives that are consistent with the arbitrage-free pricing paradigm.

Alternatively, a suitable version of the fundamental theorem of asset pricing can be used to show that the discounted prices defined through admissible trading strategies are $\sigma$-martingales (hence, in fact, supermartingales) under an equivalent local martingale measure. The latter property is a well known fundamental feature of stochastic integration, so it covers all linear market models. Obviously, our very brief summary of linear arbitrage-free pricing theory is rather superficial and we acknowledge that it should be complemented by suitable assumptions on prices of traded assets and specific definitions of no-arbitrage. For a survey of classical results regarding no-arbitrage properties of linear market models, we refer to the monograph by Delbaen and Schachermayer \cite{DS2006} (see also papers by Karatzas and Kardaras \cite{KK2007},
Kardaras \cite{K2012}, and Takaoka and Schweizer \cite{TS2014} for more recent developments).

Let us now comment on the existing approaches to the nonlinear valuation of derivatives, as first developed by
El Karoui and Quenez \cite{EQ1997} and El Karoui et al. \cite{EPQ1997}) and later applied by several authors to particular
financial models or classes on contracts (see, for instance, Bichuch et al. \cite{BCS2017}, Brigo and Pallavicini \cite{BP2014}, Cr\'epey \cite{C2015a,C2015b}, Dumitrescu et al. \cite{DQS2017}, Mercurio \cite{M2013} or Pallavicini et al. \cite{PPB2012a,PPB2012b}). The most common approach to the valuation problem in a nonlinear framework seems to hinge, at least implicitly, on the following steps in which it is usually assumed that the hedger's initial endowment is immaterial and thus it may be set to zero. In fact, Step (N.1) was explicitly addressed only in some of the above-mentioned works, whereas in most papers in the existing literature the authors were only concerned with the issue of finding a replicating or a superhedging strategy, as briefly outlined in Step (N.2). Also, to the best of our knowledge, the important issue emphasized in Step (N.3) has been completely ignored up to now, since apparently it was implicitly taken for granted
that the cost of replication, as given by a solution to a suitable BSDE, is a fair price of the contract.

\vskip 2 pt
\noindent
\textbf{Step (N.1).} The strict comparison argument for the BSDE associated with the wealth dynamics is used to show
that one cannot construct an admissible trading strategy with the null initial wealth and the terminal wealth,
which is non-negative almost surely and strictly positive with a positive probability (hence the classical no-arbitrage
property holds). \vskip 2 pt
\noindent
\textbf{Step (N.2).} The price for a European contingent claim is defined using either the cost of replication
or the minimal cost of superhedging. A suitable version of the strict comparison property for wealth processes
can be used to show that, for some nonlinear market models, the two pricing approaches yield the same value
for any replicable European claim. \vskip 2 pt
\noindent
\textbf{Step (N.3).} It remains to check if the tentative price, as given by the cost of replication or selected to be below the
upper bound given by the minimal cost of superhedging, complies with some form of the no-arbitrage property of the extended
market.

We will argue that the question whether the extended nonlinear market model preserves the no-arbitrage property
(of course, according to each particular definition of no-arbitrage is much harder to resolve than it
was the case in the linear framework. Intuitively, this is due to the fact that trading in derivatives may essentially
change the properties of the original nonlinear market, whereas some version of the FTAP can be used to give a positive answer
to the same question in the linear setup. We propose a partial solution in the nonlinear framework by putting forward
in Section \ref{sec3.9} the concept of the {\it regular} market model (see Definitions  \ref{def:de7bb} and \ref{def:de7bc})
and we establish some results on the fair pricing in a regular model (see Propositions \ref{prop3.13} and \ref{prop3.13n}).

\subsection{Discounted Wealth and Admissible Strategies}  \label{sec3.2}

To deal with the issue of no-arbitrage, we need to introduce the discounted wealth process and properly define the concept
of admissibility of trading strategies. For any $x \in \bR$, we denote by $\cB (x)$ the strictly positive process given by,
for all $t \in [0,T]$,
\begin{equation} \label{eq:discFactor}
\cB_t(x) := \I_{\{x \geq 0\}}B^{0,l}_t+\I_{\{x < 0\}}B^{0,b}_t .
\end{equation}
Note that if $B^{0,l}=B^{0,b}$, then ${\mathcal{B}} (x)=B^0=B$. Furthermore, if $x=0$, then $xB^{0,b}_t= xB^{0,l}_t=0$ for all $t \in [0,T]$
and thus the choice of either $B^{0,l}$ or $B^{0,b}$ in the right-hand side of \eqref{eq:discFactor} will be in fact immaterial.
It is natural to postulate that the initial endowment $x \geq 0$ (resp. $x<0$)
has the future value $x B^{0,l}_t$ (resp. $ xB^{0,b}_t$) at time $t \in [0,T]$ when invested in the cash account $B^{0,l}$
(resp. $B^{0,b}$). We henceforth work under the following assumption.


\begin{assumption} \label{ass3.1}
{\rm We postulate that: \vspace{-0.51em}
\begin{enumerate}[(i)]\addtolength{\itemsep}{-.5\baselineskip}
\item for any initial endowment $x \in \bR$ of the hedger, the {\it null contract} $\cN=(0,0)$ belongs to $\sC$,
\item for any $x \in \bR$, the trading strategy  $(x,0,\wh \phi,\cN )$, where all components of $\wh \phi$ vanish except for either $\psi^{0,l}$, if $x \geq 0$, or $\psi^{0,b}$,  if $x < 0$, belongs to $\Phi^{0,x}( \sC )$ and $V^p_t (x,0,\wh \phi, \cN )=V_t (x,0,\wh \phi, \cN )=x \cB_t(x) $ for all $t \in [0,T]$.
\end{enumerate} }
\end{assumption}

At the first glance, Assumption \ref{ass3.1} may look trivial or even redundant but it should be made and it will be useful in the derivation of fundamental properties of fair prices. Condition (i) is indeed a rather obvious formal requirement. Note, however, that condition (ii) cannot be deduced directly from the self-financing condition, since it hinges on the additional postulate that there are no trading adjustments (such as: taxes, transactions costs, margin account, etc.) when the initial endowment is invested in the cash account. It is needed to show that the null contract has fair price zero at any date $t \in [0,T]$. Also, the trading strategy introduced in condition (ii) will serve as a natural benchmark for assessment of profits or losses incurred by the hedger.  A natural extension of Assumption \ref{ass3.1} to the case where
we study trading strategies on $[t,T]$ is also implicitly postulated without stating it explicitly.

In the next necessary step, we follow the standard approach of introducing the concept of admissibility for the discounted wealth.
Towards this end, for any fixed $t\in [0,T)$, we consider a hedger who starts trading at time $t$ with the initial
endowment $x_t$ and uses a self-financing trading strategy $(x_t,p_t,\phi^t,\cC^t)$, where the price $p_t \in \cG_t$
at which the contract $\cC^t$ is traded at time $t$ is arbitrary. We also consider the strictly positive discounting process
$\cB^t (x_t)$, which is defined for all $u \in [t,T]$ by
\begin{equation} \label{eq:discFactc}
\cB^t_{u}(x_t) := \I_{\{x_t \geq 0\}}B^{0,l}_u (B^{0,l}_t)^{-1}+\I_{\{x_t < 0\}}B^{0,b}_u (B^{0,b}_t)^{-1},
\end{equation}
so that, in particular, $\cB^t_t (x_t)= 1$. Then the wealth process discounted back to time $t$ satisfies, for all $u\in [t,T]$,
\begin{equation} \label{edd}
\wt{V}_u(x_t,p_t,\phi^t,\cC^t):= (\mathcal{B}^t_u(x_t))^{-1} V_u(x_t,p_t,\phi^t,\cC^t)
\end{equation}
and we have the following natural concept of \textit{admissibility} of a trading strategy on $[t,T]$.

\begin{definition} \label{defadmih}
For any fixed $t\in [0,T)$, we say that a trading strategy $(x_t,p_t,\phi^t,\cC^t)\in \Phi^{t,x_t} ({\sC})$ is
{\it admissible} if the discounted wealth $\wt{V}_u(x_t,p_t,\phi^t,\cC^t)$ is bounded from below by a constant. We denote by $ \Psi^{t,x_t}(p_t,\cC^t)$ the class of admissible strategies corresponding to $(x_t,p_t,\phi^t,\cC^t)$ and we denote by
\[
\Psi^{t,x_t} ({\sC}) := \cup_{\cC  \in \sC} \cup_{p_t \in \cG_t}  \Psi^{t,x_t}(p_t,\cC^t)
\]
the class of all admissible trading strategies on $[t,T]$ relative to the class $\sC$ of contracts for the hedger with
the initial endowment $x_t$ at time $t$.  In particular, $\Psi^{0,x} ({\sC})$ is the class of all trading strategies
that are admissible for the hedger with the initial endowment $x$ at time $t=0$.
\end{definition}

\subsection{No-arbitrage with Respect to the Null Contract}  \label{sec3.3}

A minimal no-arbitrage requirement for an underlying market model is that it should be arbitrage-free with
respect to the null contract. Note that, consistently with Assumption \ref{ass3.1} and the concept
of replication (for the general formulation of replication of a non-null contract, see Definition \ref{defprice1}),
it is implicitly assumed in Definition  \ref{defarbn} that the price at which the null contract is traded at time zero equals zero.
Needless to say, this is a rather indisputable feature of any trading model.

\begin{definition} \label{defarbn0}
Consider an underlying market model $\cM=(\cS,\cD, \cB, {\sC }, \Psi^{0,x} (\sC ) )$.
An {\it arbitrage opportunity with respect to the null contract} (or a {\it primary arbitrage opportunity})  for the hedger with an initial endowment $x$ is
a strategy  $(x,0,\varphi , \cN ) \in \Psi^{0,x}(0, \cN )$  such that
\begin{equation}\label{eq:ArbUnderModel}
 \bP( \wt{V}_T(x,0,\varphi, \cN )\geq x  )=1, \quad \bP( \wt{V}_T (x,0,\varphi, \cN ) >x )>0.
\end{equation}
\end{definition}

\begin{definition} \label{defarbn}
If no primary arbitrage opportunity exists in the market model $\cM$, then we say that $\cM$ has the {\it no-arbitrage property with respect to the null contract} for the hedger with an initial endowment $x$.
\end{definition}

For an arbitrary linear market model, Definition \ref{defarbn} reduces to the classical definition of an arbitrage opportunity.
It is well known that the no-arbitrage property introduced in this definition is a sufficiently strong tool for the development
of arbitrage-free pricing for financial derivatives in the linear framework. This does not mean, however, that Definition \ref{defarbn} is sufficiently strong to allow us to develop nonlinear arbitrage-free pricing theory, which would enjoy the properties, which are desirable from either mathematical or financial perspective.

 On the one hand, a natural definition of a {\it hedger's fair value} (see Definition \ref{def:deffi7b}) is consistent with the concept of no-arbitrage with respect to the null contract and thus Definition \ref{defarbn} seems to be theoretically sound. On the other hand, however, Definition \ref{defarbn} is inadequate for an efficient valuation and hedging approach in a general nonlinear market for the following reasons. First, it may occur that the replication cost of a contract does not satisfy the definition of a fair price, since the possibility of selling of a contract at the hedger's replication cost may generate an arbitrage opportunity for him. An explicit example of a market model, which is arbitrage-free in the sense of Definition~\ref{defarbn}, but suffers from this deficiency, is given in Section~\ref{sec4.2.3}. Second, and perhaps even more importantly, no well established method for finding a fair price in a general nonlinear market satisfying Definition \ref{defarbn} is available.

We contend that the drawback of the definition of an arbitrage-free model with respect to the null contract is that it does not make an explicit reference to a class $\sC$ of contracts under study. Indeed, it relies on the specification of the class $\Psi^{0,x}(0, \cN )$ of trading strategies, but it makes no reference to the larger class $\Psi^{0,x}(\sC)$. To amend that drawback of Definition \ref{defarbn}, Bielecki and Rutkowski \cite{BR2015} proposed to consider the concept of the no-arbitrage property for the trading desk with respect to a predetermined family $\sC$ of contracts.

\subsection{No-arbitrage for the Trading Desk}    \label{sec3.4}

Following Bielecki and Rutkowski \cite{BR2015}, we will now examine a stronger no-arbitrage property of a market model, which is intimately related to a predetermined family $\sC$ of financial contracts. Our goal is here to propose
a more stringent no-arbitrage condition, which not only accounts for the nonlinearity of the market, but also explicitly refers to a family of contracts under consideration. Unfortunately, the class of models that are arbitrage-free in the sense of Definition  \ref{defarbt} seems to be too encompassing and thus it is still unclear whether valuation irregularities commented upon in the preceding section will be completely eliminated (for an example, see Section \ref{sec4.2.3}).

For simplicity of notation, we consider here the case of $t=0$, but all definitions can easily be extended
to the case of any date $t$. The symbols $\cX=\cX (A)$ and $\cY=\cY (- A)$ are used to emphasize that there is
no reason to expect that the trading adjustments will satisfy the equality $\cX (-A)= -\cX (A)$, in general.
Therefore, we denote by $\cY=(Y^1, \dots , Y^n; \alpha^1(\cY), \ldots, \alpha^{n}(\cY); \beta^1(\cY), \ldots, \beta^n(\cY))$
the trading adjustments associated with the cumulative cash flows process $-A$. In order to avoid confusion, we will use the full notation for the wealth process, for instance, $V(x,p,\phi,\cC)=V(x,p,\phi, A,\cX)$, etc.

\begin{remark}
As already mentioned above, it is not necessarily true that the equality $Y^k=-X^k$, holds for all $k=1,2,\ldots,n$.
For instance, this equality is satisfied by the variation margin, but it is not met by the initial margin and the regulatory capital, which are always non-negative.
\end{remark}

\begin{definition} \label{defcomwea}
For a contract $\cC=(A,\cX)$ and an initial endowment $x$, the {\it combined wealth} is defined as
\begin{equation} \label{comwea}
\Vcom (x_1, x_2,\phi, \bar \phi, A,\cX,\cY) := V(x_1,0,\phi, A,\cX)+V(x_2,0, \bar \phi, -A, \cY),
\end{equation}
where  $x_1, x_2$ are arbitrary real numbers such that $x=x_1+x_2,\phi \in \Psi^{0,x_1}(0,A,\cX),$
and $\, \bar \phi \in \Psi^{0,x_2}(0,-A, \cY)$. In particular, $\Vcom_0(x_1, x_2,\phi, \bar \phi, A,\cX,\cY)= x_1 +x_2=x$.
\end{definition}

The motivation for the name {\it combined wealth} is fairly transparent, since it comes directly from the financial interpretation of the process given by the right-hand side in \eqref{comwea}. We argue below that it can be seen as the aggregated wealth of the two traders, who are members of the same trading desk, and who are supposed to proceed as follows:

\begin{itemize}\addtolength{\itemsep}{-.5\baselineskip}
\item The first trader takes the long position in a contract $(A,\cX)$, whereas the second one takes the short position in the same contract, so that his position is formally represented by $(-A, \cY)$. Since we assume that the long and short
positions have exactly opposite prices, the corresponding cash flows $p$ and $-p$ coming to the trading desk (and not to individual
traders) offset each other and thus the initial endowment $x$ of the trading desk remains unchanged.
\item In addition, it is assumed that after the cash flows $p$ and $-p$ have already been netted, so they are no longer relevant, the initial
endowment $x$ is split into arbitrary amounts $x_1$ and $x_2$ meaning that $x=x_1+x_2$. Then each trader is allocated the respective amount $x_1$ or $x_2$ as his initial endowment and each of them undertakes active hedging of his respective position. It is now clear that the level of the initial price $p$ at which the contract is traded at time zero is immaterial for both hedging strategies and the total (i.e., combined) wealth of the two traders is given by the right-hand side in \eqref{comwea}.
\end{itemize}

Alternatively, the combined wealth may be used to describe the situation where a single trader takes long and short positions with two external counterparties and hedges them independently using his initial endowment $x$ split into $x_1$ and $x_2$. Of course, in that
case it is even more clear that the initial price $p$ does not affect his trading strategies since the amount of cash received at time
0 from one of the counterparties is immediately transferred to the second one.

\begin{remark}
One can also observe that the following equality holds for any real number $p$
\[
V(x_1,0,\phi, A,\cX)+V(x_2,0, \bar \phi, -A, \cY)=V(\wt x_1,p,\phi, A,\cX)+V(\wt x_2,-p, \bar \phi, -A, \cY),
\]
where $\wt x_1=x_1-p$ and $\wt x_2=x_2+p$ is another decomposition of $x$ such that $x=\wt x_1+\wt x_2$.
However, equation \eqref{comwea} better reflects the actual trading arrangements and it has a clear advantage that a
number $p$, which is not known a priori, does not appear in the expression for the combined wealth. Hence \eqref{comwea} emphasizes the crucial feature that the combined wealth is independent of $p$. In fact, one can remark that the fact whether the trading desk is aware about the actual level of the price $p$ is immaterial for the question whether an arbitrage opportunity for the trading desk exists in
a particular market model.
\end{remark}

\begin{definition} \label{defiadmiw}
A pair $(x_1,\phi; x_2, \bar \phi)$ of trading strategies introduced in Definition \ref{defcomwea}
is {\it admissible for the trading desk} if  the discounted combined wealth process
\begin{equation} \label{ycomwea}
\Vcomt (x_1, x_2 ,\phi, \bar \phi, A,\cX , \cY ):=(\mathcal{B}(x))^{-1}\Vcom (x_1, x_2,\phi,\bar \phi,A,\cX , \cY )
\end{equation}
is bounded from below by a constant.  The class of strategies admissible for the trading desk is denoted by $\Psi^{0,x_1,x_2}(A,\cX , \cY)$.
\end{definition}

We are in a position to formalize the concept of an arbitrage-free model for the trading desk with respect to
a particular family of contracts.

\begin{definition} \label{defarbt0}
 A pair $(x_1,\phi; x_2,\bar \phi) \in \Psi^{0,x_1,x_2}(A,\cX , \cY)$ is an {\it arbitrage opportunity for the trading desk} with respect to a contract $(A,\cX)$ if the following conditions are satisfied
\[
\bP ( \Vcomt_T  (x_1, x_2,\phi, \bar \phi,A,\cX ,\cY ) \geq x )=1,
\quad \bP (\Vcomt_T  (x_1, x_2,\phi,\bar \phi,A,\cX,\cY ) > x ) > 0 .
\]
\end{definition}

\begin{definition} \label{defarbt}
We say that the market model $\cM=(\cS,\cD, \cB, \sC , \Psi^{0,x} (\sC ) )$ has the {\it no-arbitrage property for the trading desk} if
there are no arbitrage opportunities for the trading desk with respect to any contract $\cC $ from $\sC $.
\end{definition}

Our main purpose in Sections \ref{sec3.3} and \ref{sec3.4} was to provide some simple and financially meaningful criteria that would allow us to detect and eliminate market models in which some particular form of arbitrage appears. Definition \ref{defarbn} and  Definition \ref{defarbt} provide such criteria for accepting or rejecting any tentative nonlinear market model. It is easy to see that a model which is rejected according to  Definition \ref{defarbt} will also be rejected if Definition \ref{defarbn} is applied. We do not claim, however, that these tentative tests are sufficient for an effective discrimination between acceptable and non-acceptable nonlinear models for valuation of derivatives. Therefore, in Definition \ref{def:de7bb}, we will formulate additional conditions that should be satisfied by an acceptable model, which is then called a {\it regular} model.

\subsection{Dynamics of the Discounted Wealth Process}    \label{sec3.5}

It is natural to ask whether the no-arbitrage for the trading desk can be checked for
a given market model. Before we illustrate a simple verification method for this property,
we need to introduce additional notation. Let us write
\begin{align*}
&\wt B^{i,l}(x) := (\cB (x))^{-1} B^{i,l}, \quad  \quad \quad \wt B^{i,b}(x) := (\cB (x))^{-1} B^{i,b}, \\
&\wt \beta^k (x,\cX) := (\cB(x))^{-1}\beta^k (\cX) , \quad \wt {\beta}^k (x,\cY) := (\cB(x))^{-1}\beta^k(\cY), \\
&\wh X^k := (\beta^k (\cX) )^{-1}X^k,\quad \quad \quad \quad \wh  Y^k := (\beta^k(\cY))^{-1} Y^k, \\
&B^{0,b,l} := (B^{0,l})^{-1}B^{0,b}, \quad \quad \quad \ \ B^{0,l,b} := (B^{0,b})^{-1}B^{0,l}.
\end{align*}

\begin{lemma}
The discounted combined wealth satisfies
\begin{equation}
\begin{aligned} \label{eq:dynwtVcom}
 & d  \Vcomt_t(x_1,x_2,\varphi, \bar \varphi, A,\cX, \cY )
= \sum_{i=1}^{d}  (\xi_t^i +\bar \xi_t^i)\,d \wt S^{i,cld}_t(x)
+\sum_{i=1}^{d}  (\psi_t^{i,l}+\bar \psi_t^{i,l} )\,  d\wt B_t^{i,l}(x)  \\
 &+\sum_{i=1}^{d}  (\psi_t^{i,b}+\bar \psi_t^{i,b})\,d\wt B_t^{i,b}(x) +\I_{\{x\geq 0\}} (\psi_t^{0,b}+\bar  \psi_t^{0,b} )\, dB_t^{0,b,l}+ \I_{\{x<0\}} (\psi_t^{0,l}+\bar \psi_t^{0,l})\,dB_t^{0,l,b}   \\
 & - \sum_{k=1}^{n} \wh X_t^k\,d \wt \beta_t^k (x,\cX ) -\sum_{k=1}^{n} \wh Y_t^k\, d \wt {{\beta}}_t^k (x,\cY)\\
 &+ \sum_{k=1}^{n} \big((1-\alpha_t^k (\cX) ) X^k_t+(1- \alpha_t^k(\cY)) Y^k_t \big)\,d(\cB_t(x))^{-1},
\end{aligned}
\end{equation}
where we set
\begin{equation} \label{sicld}
\wt S^{i,cld}_t(x) := (\cB_t(x))^{-1}S_t^i+\int_0^t (\cB_u(x))^{-1}\, dD^i_u .
\end{equation}
\end{lemma}

\begin{proof} For an arbitrary decomposition $x=x_1+x_2$,  we write (note that the notation
introduced in equation \eqref{edd} is extended here, since $x \ne x_i$, in general)
\begin{align*}
&\wt V(x_1,p,\varphi,A,\cX) := (\cB (x))^{-1} \wt V(x_1,p,\varphi,A,\cX), \\
&\wt V(x_2,p,\bar \varphi,-A,\cY) := (\cB (x))^{-1} \wt V(x_2,p,\bar \varphi,-A,\cY).
\end{align*}
From \eqref{wealth2} and \eqref{wealth3-1}, using the It\^o integration by parts formula, we obtain
\begin{align} \label{wealthdyna}
 d\wt V_t&(x_1,p,\varphi,A,\cX)=  \sum_{i=1}^{d}  \xi_t^i\,d \wt S^{i,cld}_t(x) +
 \sum_{i=1}^{d} \left(\psi_t^{i,l} d\wt B_t^{i,l}(x)+\psi_t^{i,b}\, d\wt B_t^{i,b} (x) \right) \nonumber \\
 &+\I_{\{x\geq 0\}} \psi_t^{0,b}\, dB_t^{0,b,l}+\I_{\{x<0\}}\psi_t^{0,l}\, dB_t^{0,l,b}+(\cB_t(x))^{-1}\, dA_t-\sum_{k=1}^{n} \wh X_t^k\,d \wt \beta_t^k (x,\cX ) \\
 &+\sum_{k=1}^{n} (1-\alpha_t^k (\cX) ) X^k_t\,d (\cB_t(x))^{-1}, \nonumber
\end{align}
and an analogous equality holds for $\wt V(x_2,p,\bar \varphi,-A,\cY)$. Hence \eqref{eq:dynwtVcom}
follows from \eqref{comwea} and \eqref{ycomwea}.
\end{proof}

We deduce from \eqref{wealthdyna} that condition (ii) in  Assumption \ref{ass3.1} is satisfied, provided that no additional portfolio
constraints are imposed (recall that condition (i) in Assumption \ref{ass3.1} is always postulated to hold).

Assume now, in addition, that $B^{i,l}=B^{i,b}=B^i$ for $i=1,2,\ldots, d$. We define the processes
$S^{i,\textrm{cld}}$ and $\wh S^{i,\textrm{cld}}$
\begin{equation*}
S^{i,\textrm{cld}}_t := S^i_t+ B^{i}_t \int_0^t (B^i_u)^{-1}\,d D^i_u , \quad
\wh S^{i,\textrm{cld}}_t :=  (B^i_t)^{-1} S_t^{i,\textrm{cld}}=\wh S^i_t+ \int_0^t (B^i_u)^{-1}\,d D^i_u,
\end{equation*}
where in turn $\wh S^i :=  (B^i)^{-1} S^i$. It is easy to check that
\begin{equation} \label{newdd}
d \wt S_t^{i,cld} (x)=\wt B^i_t(x) \,d \widehat S_t^{i,cld}+\wh S_t^i\,d\wt B^i_t(x),
\end{equation}
where  $\wt B^i (x) := (\cB(x))^{-1} B^i$.

\begin{corollary}
Assume that $B^{i,l}=B^{i,b}=B^i$ for $i=1,2,\ldots, d$. Then the discounted combined wealth satisfies
\begin{align} \label{newnew}
d \Vcomt_t&(x_1,x_2,\phi, \bar \phi, A,\cX , \cY )
=  \sum_{i=1}^{d} (\xi_t^i+\bar \xi^i_t) \wt B_t^i(x)\, d \widehat S_t^{i,cld} +\sum_{i=1}^{d}\big( (\xi_t^i+\bar \xi^i_t) \wh S_t^i+(\psi_t^{i}+ \bar \psi^i_t) \big)\,d\wt B_t^i (x) \nonumber
\\ &+\I_{\{x\geq 0\}} (\psi_t^{0,b}+\bar  \psi_t^{0,b} )\, dB_t^{0,b,l}+\I_{\{x<0\}} (\psi_t^{0,l}+\bar \psi_t^{0,l})\, dB_t^{0,l,b} - \sum_{k=1}^{n} \wh X_t^k d \wt \beta_t^k (x,\cX ) \\
& -\sum_{k=1}^{n} \wh Y_t^k\, d \wt {\beta}_t^k (x,\cY)   +\sum_{k=1}^{n} \big( (1-\alpha_t^k (\cX) ) X^k_t+(1- \alpha_t^k(\cY)) Y^k_t \big)\,d(\cB_t(x))^{-1}. \nonumber
\end{align}
\end{corollary}

\proof
It suffices to combine \eqref{eq:dynwtVcom} with \eqref{newdd}.
\endproof

\subsection{Sufficient Conditions for the Trading Desk No-Arbitrage}  \label{sec3.6}

The following result gives a sufficient condition for a market model to be arbitrage-free for the trading desk.
The proof of Proposition \ref{propoarbi1} is pretty straightforward and thus it is omitted.

\begin{proposition} \label{propoarbi1}
Assume that there exists a probability measure $\bQ$, equivalent to $\bP$ on $(\Omega , \cG_T)$, and such that for any decomposition $x=x_1+x_2$  and any admissible combination of trading strategies $(x_1 ,\phi,A, \mathcal{X})$ and $(x_2, \bar \phi, -A, \mathcal{Y})$ for any contract $(A,\cX )$ belonging to $\sC $ the discounted combined wealth $\Vcomt(x_1,x_2,\phi, \bar \phi, A,\cX, \cY )$  is a supermartingale under $\bQ$. Then the market model $\cM=(\cS,\cD, \cB, \sC , \Psi^{0,x} (\sC ) )$ is arbitrage-free for the trading desk.
\end{proposition}

Although Proposition~\ref{propoarbi1} is fairly abstract, the sufficient condition formulated there can readily be verified, as soon as a specific market model is adopted (see, for instance, Bielecki and Rutkowski \cite{BR2015} and Nie and Rutkowski \cite{NR2015,NR2016a,NR2017}). To support this claim, we will examine an example of a market model with idiosyncratic funding for risky assets and rehypothecated cash collateral.

\begin{example}  \label{ex1}
We consider the special case where $B^{0,l}=B^{0,b}=B=\cB (x)$ and $B^{i,l}=B^{i,b}=B^i$ for all $i=1,2,\ldots, d$.
If we temporarily assume that there are no additional portfolio constraints, then from
\eqref{newnew} we obtain (for a special case of this formula, see Corollary 2.1 in Bielecki and Rutkowski \cite{BR2015})
\begin{align*}
d \Vcomt_t &(x_1,x_2,\phi, \bar \phi, A,\cX , \cY )  =\sum_{i=1}^{d} (\xi_t^i+\bar \xi^i_t) \widetilde B_t^i(x)\,d \widehat{S}_t^{i,\textrm{cld}}+\sum_{i=1}^{d} ( \xi_t^i S_t^i+\psi_t^i B^i_t) (B^i_t)^{-1}\,d\widetilde B_t^i(x) \nonumber \\  &+\sum_{i=1}^{d} ( \bar \xi_t^i S_t^i+\bar \psi_t^i B^i_t) (B^i_t)^{-1} \,  d\widetilde B_t^i(x)
- \sum_{k=1}^{n}  \wh X_t^k \,d\widetilde{\beta}_t^k(x,\cX )-\sum_{k=1}^{n} \wh Y_t^k\, d\widetilde{\beta}_t^k(x,\cY)  \\
&+\sum_{k=1}^{n} \big( (1-\alpha_t^k (\cX) ) X^k_t+ (1- \alpha_t^k(\cY)) Y^k_t \big)\, dB_t^{-1}.
\end{align*}
We postulate that the cash collateral is rehypothecated, so that $n_1=n=2$ in Lemma \ref{lmnb}. Then $\alpha^1_t=\alpha^2_t
=\alpha^1_t(\cY) =\alpha^2_t(\cY)=1$ and $X^1_t+Y^1_t=X^2_t+Y^2_t=0$ for all $t \in [0,T]$. Let us assume, in addition, that $\xi_t^i S_t^i+ \psi_t^i B^i_t=\bar \xi_t^i S_t^i+\bar \psi_t^i B^i_t =0$ for all $i$ and $t \in [0,T]$, which means that the $i$th risky asset is fully funded using the repo account $B^i$ (see Section \ref{sec2.6.3}). More generally, it suffices to assume that the following equalities are satisfied for all $t \in [0,T]$
\begin{equation} \label{conee}
\sum_{i=1}^d \int_0^t ( \xi^i_u S^{i}_u+\psi^i_u B^i_u ) (B^i_u)^{-1}\,d\wt B^i_u(x)
= \sum_{i=1}^d \int_0^t (\bar \xi^i_u S^{i}_u+\bar \psi^i_u B^i_u )(B^i_u)^{-1} \,d\wt B^i_u(x)=0.
\end{equation}
Finally, let the remuneration processes satisfy $\beta^k(\cX) =\beta^k(\cY)$ (the symmetry of collateral rates).
Then the formula for the dynamics of the discounted combined wealth for the trading desk reduces to
$$
d \Vcomt_t (x_1, x_2,\phi, \bar \phi, A,\cX , \cY )  = \sum_{i=1}^{d} (\xi_t^i+\bar \xi^i_t) \widetilde B_t^i(x)\, d \widehat{S}_t^{i,\textrm{cld}}
$$
and thus the model is arbitrage-free for the trading desk provided that there exists a probability measure $\bQ$, which is equivalent to $\bP$ on $(\Omega , \cG_T)$, and such that the processes $\widehat{S}^{i,\textrm{cld}},\, i=1,2,\ldots, d$ are $\bQ$-local martingales.  This property is still a sufficient condition for the trading desk no-arbitrage when the cash account $B^{0,l}$ and $B^{0,b}$ differ, but the borrowing rate dominates the lending rate.
\end{example}

\section{Hedger's Fair Pricing and Market Regularity}  \label{sec3.7}

We now address the issue of a fair pricing in the nonlinear framework under the assumption that the hedger has the initial endowment $x_t$ at time $t$. We assume that the model enjoys the no-arbitrage property either with respect to the null contract or for the trading desk and we consider the hedger who contemplates entering into the contract $\cC^t $ at time~$t$. The first goal is to describe the range of the {\it hedger's fair prices} for the contract $\cC^t$. Let $p_t \in \cG_t$ denote a generic price of a contract at time $t$, as seen from the perspective of the hedger. Hence if $p_t$ is positive, then the hedger receives at time $t$ the cash amount $p_t$  from the counterparty, whereas a negative value of $p_t$ means that he agrees to pay the cash amount $-p_t$ to the counterparty at time~$t$.  In the next definition, we fix a date $t \in [0,T)$ and we assume that the contract $\cC^t$ is traded at the price $p_t$ at time $t$. It is natural to ask whether in this situation the hedger can make a risk-free profit by entering into the contract and hedging it with an admissible trading strategy over $[t,T]$. We propose to call it a {\it hedger's pricing} arbitrage opportunity. Recall that the arbitrage opportunities defined in Section~\ref{sec3} are related to the properties of a trading model, but they do not depend on the level of a price $p_t$ for $\cC^t $.

\begin{definition} \label{def:deffi7b}
A trading strategy  $(x_t,p_t,\phi^t,\cC^t)\in \Psi^{t,x_t} (\sC )$ is a \textit{hedger's pricing
arbitrage opportunity on $[t,T]$} associated with a contract $\cC^t $ traded at $p_t$ at time $t$
(or, briefly, a {\it secondary arbitrage opportunity}) if
\begin{equation} \label{price1b}
\bP \big( \widetilde{V}_T(x_t,p_t,\phi^t,\cC^t) \geq x_t \big)=1
\end{equation}
and
\begin{equation}\label{price2b}
\bP \big( \widetilde{V}_T(x_t,p_t,\phi^t,\cC^t) > x_t \big) > 0.
\end{equation}
\end{definition}

It is clear that $(x_t,p_t,\phi^t,\cC^t)\in \Psi^{t,x_t} (\sC )$ is not a hedger's pricing
arbitrage opportunity on $[t,T]$ if either
\begin{equation} \label{price1v1}
\bP \big( \widetilde{V}_T(x_t,p_t,\phi^t,\cC^t)=x_t \big)=1
\end{equation}
or
\begin{equation} \label{price1v2}
\bP \big( \widetilde{V}_T(x_t,p_t,\phi^t,\cC^t) < x_t \big) > 0 .
\end{equation}
We will refer to condition \eqref{price1v2} as the hedger's {\it loss condition}.

\begin{definition}  \label{deffairp}
We say that $p^f_t=p^f_t(x_t,\cC^t)$ is a \textit{fair hedger's price} at time $t$ for $\cC^t$ if there is
no hedger's secondary arbitrage opportunity $(x_t,p^f_t,\phi^t,\cC^t)\in \Psi^{t, x_t}(\sC )$. A fair hedger's price $p^f_t$ such that the loss  condition holds for every trading strategy $(x_t,p^f_t,\phi^t,\cC^t)\in \Psi^{t,x_t} (\sC )$ is called a {\it loss-generating cost.}
\end{definition}

It is clear from Definitions \ref{def:deffi7b} and \ref{deffairp} that if $p^f_t$ is a fair price,
then any trading strategy $(x_t,p_t,\phi^t,\cC^t)\in \Psi^{t,x_t} (\sC )$ necessarily satisfies either condition \eqref{price1v1}  or condition \eqref{price1v2}. Obviously, a fair hedger's price depends both on the given endowment $x_t$ and the contract $\cC^t$ and thus the notation $p^f_t(x_t,\cC^t)$ is appropriate, but it will be frequently simplified to $p^f_t$ when no danger of confusion may arise.

A fair price prevents the hedger from making a sure profit with a positive probability and with no risk of suffering a loss.
In contrast, it does not prevent the hedger from losing money, in general. It is thus natural to search for the highest
level of a fair price. This idea motivates the following definition of the upper bound for the hedger's fair prices
$ \ovep^f_t(x_t,\cC^t)=\esssup \cH^f_t(x_t,\cC^t)$, where
\begin{align} \label{eq:FairPrice1}
\cH^f_t(x_t,\cC^t):=\big\{ p^f_t\in \cG_t \mid p^f_t \textrm{\ is a fair hedger's price for\ }\cC^t\big\}.
\end{align}
We find it useful to study also the upper bound for the loss-generating  costs, which is given by
$ \ovepl_t(x_t,\cC^t)=\esssup \cH^l_t(x_t,\cC^t)$ where
\begin{align} \label{eq:FairPrice1v}
\cH^l_t(x_t,\cC^t):=\big\{ p^l_t\in \cG_t \mid p^l_t \textrm{\ is a loss-generating  cost for\ }\cC^t\big\}.
\end{align}

\begin{definition} \label{defsupcost}
A trading strategy $(x_t,p^s_t,\phi^t,\cC^t)\in \Psi^{t,x_t} (\sC )$ is said to be the {\it superhedging strategy}
for the contract $\cC^t$ if \eqref{price1b} holds, whereas the {\it strict superhedging} means that \eqref{price1b} and \eqref{price2b} are satisfied. If $p^s_t=p^s_t(x_t,\cC^t)$ is such that there exists a superhedging strategy  (resp. a strict superhedging strategy) $(x_t,p^s_t,\phi^t,\cC^t)\in \Psi^{t,x_t}(\sC)$, then it is  called a \textit{superhedging cost} (resp. \textit{strict superhedging cost}) at time $t$ for $\cC^t$.
\end{definition}

The lower bound for superhedging costs is given by  ${\undp}^s_t (x_t,\cC^t)=\essinf  \cH^s_t (x_t,\cC^t)$ where
\begin{align} \label{eq:FairPrice2}
\cH^s_t(x_t,\cC^t):=\big\{ p^s_t\in \cG_t \mid p^s_t \textrm{\ is a superhedging cost for\ }\cC^t \big\}
\end{align}
and the lower bound for strict superhedging costs equals $\undpa_t (x_t,\cC^t)=\essinf  \cH^a_t (x_t,\cC^t)$ where
\begin{align} \label{eq:FairPrice2v}
\cH^a_t(x_t,\cC^t):=\big\{ p^a_t\in \cG_t \mid p^a_t \textrm{\ is a strict superhedging cost for\ }\cC^t \big\}.
\end{align}

From Definitions \ref{deffairp} and \ref{defsupcost}, it is obvious that $\cH^l_t (x_t,\cC^t) \subseteq \cH^f_t (x_t,\cC^t)$ and $\cH^a_t (x_t,\cC^t) \subseteq \cH^s_t (x_t,\cC^t)$. Moreover, the set $\cH^f_t (x_t,\cC^t)$ is the complement of $\cH^a_t (x_t,\cC^t)$ and the set $\cH^l_t (x_t,\cC^t)$ is the complement of $\cH^s_t (x_t,\cC^t)$ so that, obviously, we have that $\cH^f_t (x_t,\cC^t) \cap \cH^a_t (x_t,\cC^t)=\emptyset $ and $\cH^l_t (x_t,\cC^t) \cap \cH^s_t (x_t,\cC^t)=\emptyset $. Consequently, it is easy to see that
\begin{equation} \label{eqlem1v}
\ovepl_0(x,\cC) \leq {\ovep}^f_0(x,\cC), \quad \undp^s_0(x,\cC) \leq \undpa_0(x,\cC),
\end{equation}
where, in principle, it may happen that $\undp^s_0(x,\cC)=- \infty $ or ${\ovep}^f_0(x,\cC)=\infty .$

The next assumption looks fairly natural, but it is not necessarily satisfied by every nonlinear market model, so it should be checked on a case-by-case basis.

\begin{assumption} \label{ass1.1w}
{\rm  For every $\cC \in \sC$ and $t \in [0,T)$, all $x_t,p_t,q_t \in \mathcal{G}_t$, and every trading strategy $(x_t,p_t,\phi^t,\cC^t)\in \Psi^{t,x_t}(\sC)$, if $q_t \geq p_t$ on some event $D \in \mathcal{G}_t$  such that $\mathbb{P}(D)>0$, then there exists a trading strategy $(x_t,q_t,\psi^t,\cC^t)\in \Psi^{t,x_t} (\sC )$ such that the inequality $V_T(x_t,q_t,\psi^t,\cC^t) \geq V_T(x_t,p_t,\phi^t,\cC^t)$ holds on $D$.
} \end{assumption}

One may try to argue that if the hedger can enter into a contract at a higher price, then he can invest the cash difference $q_t-p_t$ in the bank account till the contract's maturity and trade according to the trading strategy $\varphi$ corresponding to $p_t$, and thus yielding a higher terminal wealth. However, this is not necessarily possible, since it may happen that $\varphi$ requires borrowing from the bank account at some future times, while we postulated that simultaneous borrowing and lending of cash is prohibited in our model. More generally, selling at a higher price means essentially that the hedger starts with a different initial capital, which may change the trading strategy significantly. Formally, a simple combination of two self-financing strategies is no longer a self-financing strategy, in general. Furthermore, if we take into account the limited supply of investment grade assets (perhaps also inclusive of the bank account), then any strategy needs to satisfy suitable portfolio constraints, which could imply that the additional cash amount received by the hedger will necessarily be used to purchase assets with a much higher exposure to substantial losses. To sum up, due to the presence of portfolio constraints and trading adjustments, the monotonicity of the terminal wealth with respect to the price $p_t$ is by no means guaranteed, in general.

Assumption~\ref{ass1.1w} is not explicitly stated in most existing papers devoted to the nonlinear valuation of derivatives. Note, however, that if the wealth process happens to be governed by some simple dynamics with no portfolio constraints or trading adjustments, then there is no need to postulate that this property holds, since it can be deduced from a suitable comparison theorem for ordinary differential equations. For a particular example of a model where Assumption~\ref{ass1.1w} is satisfied, see Lemma 6.2 in the paper by Dumitrescu et al. \cite{DQS2017} who postulate that the wealth process satisfies
\[
dV_t =-g(t,V_t,Z_t)\, dt+Z^1_t\, dW_t+Z^2_t\, dM_t,
\]
where $W$ is the Wiener process and $M$ is the compensated martingale of the default indicator process. Let us stress that if the validity of Assumption \ref{ass1.1w} is not ensured, then one may only claim that the inequalities given in \eqref{eqlem1v} are satisfied. In contrast, if Assumption \ref{ass1.1w} is met, then one obtains more informative conditions \eqref{eqlem1}.

Under Assumption \ref{ass1.1w}, it is easy to see that if $p^f_t$ is a fair hedger's price (resp. a loss-generating  cost) at time $t$ for $\cC^t$, then any $p_t \in \mathcal{G}_t$ such that $p_t\leq p^f_t$ is also a fair hedger's price  (resp. a loss-generating  cost) at time $t$ for that contract. Similarly, if a superhedging (resp. strict superhedging) strategy exists when $\cC^t$ is entered into at the price $p^s_t$,
then a superhedging  (resp. strict superhedging) strategy exists as well for any $p_t$ satisfying $ p_t \geq p^s_t$.
Therefore, if we postulate that Assumption \ref{ass1.1w} is met, then we have the following result in which we focus on the case where $t=0$, but it is easy to check that analogous properties are valid for any $t \in (0,T)$ as well.

\begin{lemma} \label{lemfun}
Let Assumption \ref{ass1.1w} be satisfied. Then for every contract $\cC \in \sC$ such that $\undpa_0(x,\cC) > - \infty $
and ${\ovep}^f_0(x,\cC) < \infty $, we obtain the following intervals:\vspace{-0.5em}
\begin{enumerate}[(i)]\addtolength{\itemsep}{-.5\baselineskip}
\item the interval $I_0^l(x,\cC):=(-\infty, \ovepl_0(x,\cC) )$ of loss generating costs,
\item the interval $I_0^r(x,\cC):=(\ovepl_0(x,\cC),\undpa_0(x,\cC))=(\undp^s_0(x,\cC),{\ovep}^f_0(x,\cC))$ where for every $p \in I_0^r(x,\cC)$
 there exists a trading strategy $(x,p,\phi,\cC)\in \Psi^{t,x}(\sC )$ such that $\bP(\widetilde  V_T(x, p,\phi,\cC)= x)=1$,
\item the interval $I_0^a(x,\cC):=(\undpa_0(x,\cC), +\infty)$ of strict superhedging costs (arbitrage range),
\end{enumerate}
where, in particular, the interval $I_0^r(x,\cC)$ may be empty. Consequently,
\begin{equation} \label{eqlem1}
\ovepl_0(x,\cC)=\undp^s_0(x,\cC) \leq {\ovep}^f_0(x,\cC)=\undpa_0(x,\cC).
\end{equation}
\end{lemma}

\begin{proof}
Under Assumption \ref{ass1.1w}, if $p \in \cH^l_0(x,\cC)$ (resp. $p \in \cH^f_0(x,\cC)$) and $q<p$, then $q \in \cH^l_0(x,\cC)$
(resp. $q \in \cH^f_0(x,\cC)$). Also, if $p \in \cH^s_0(x,\cC)$  (resp. $p \in \cH^s_0(x,\cC)$) and $q>p$, then $q\in \cH^s_0(x,\cC)$ (resp. $q \in \cH^a_0(x,\cC)$). Using, in particular, it is now easy to see that the asserted properties are valid.
\end{proof}

Observe that nothing specific can be said about the end points of the three intervals introduced in Lemma~\ref{lemfun}, in general.
Obviously, we have that either (a) $I_0^r(x,\cC)=\emptyset$ and the equalities $\ovepl_0(x,\cC)={\ovep}^f_0(x,\cC)=\undp^s_0(x,\cC)=\undpa_0(x,\cC)$ hold  or (b) $I_0^r(x,\cC)\ne \emptyset$ and thus $\undp^s_0(x,\cC)<{\ovep}^f_0(x,\cC)$. One may also consider the following stronger version of Assumption~\ref{ass1.1w} under which case (b) cannot occur, by virtue of the postulated strict monotonicity of the terminal wealth with respect to the price $p_t$.

\begin{assumption} \label{ass1.1s}
{\rm For every $\cC \in \sC$ and $t \in [0,T)$, all $x_t, p_t , q_t \in \mathcal{G}_t$, and every trading strategy $(x_t,p_t,\phi^t,\cC^t)\in \Psi^{t,x_t} (\sC )$, if $q_t > p_t$ on some event $D \in  \mathcal{G}_t$ such that $\mathbb{P}(D)>0$, then there exists a trading strategy $(x_t, q_t,\psi^t ,\cC^t)\in \Psi^{t,x_t}(\sC )$ such that the inequalities $V_T(x_t, q_t , \psi^t ,\cC^t) \geq V_T(x_t, p_t,\phi^t,\cC^t)$ and $V_T(x_t, q_t , \psi^t ,\cC^t) \ne V_T(x_t,p_t,\phi^t,\cC^t)$ are valid on $D$.
} \end{assumption}

It is clear that, under Assumption \ref{ass1.1s}, the equality $\undp^s_0(x,\cC)=\undpa_0(x,\cC)$  holds and thus the following corollary to Lemma \ref{lemfun} is valid where, by convention, $\inf \emptyset=\infty$ and $\sup \emptyset=- \infty $.

\begin{lemma}  \label{lemfunx}
If Assumption \ref{ass1.1s} is satisfied, then for any contract $\cC \in \sC$ we have that
\begin{equation} \label{eqlem2}
\ovepl_0(x,\cC)={\ovep}^f_0(x,\cC)=\undp^s_0(x,\cC)=\undpa_0(x,\cC).
\end{equation}
\end{lemma}

\subsection{Replication on $[0,T]$ and the Gained Value}         \label{sec3.8}

Admittedly, the most commonly used technique for valuation of derivatives hinges on the concept of replication.
In the present framework, it is given by the following definition, in which we consider the hedger with
the initial endowment $x$ at time 0 and where $p^r_0$ stands for an arbitrary real number.

\begin{definition} \label{defprice1}
A trading strategy $(x,p^r_0,\phi,\cC) \in \Psi^{0,x}(\sC)$ is said to {\it replicate} the contract $\cC $ on $[0,T]$ whenever $V_T(x,p^r_0,\phi,\cC)=x\cB_T(x)$ or, equivalently, $\wt{V}_T(x,p^r_0,\phi,\cC)=x$.
Then a real number $p^r_0= p^r_0(x,\cC) $ is said to be a {\it hedger's replication cost} for $\cC $ at time $0$ and
the process $p^g (x,\cC)$ given by
\begin{equation} \label{xprice1}
p^g_t(x,\cC):= V_t(x,p^r_0(x,\cC),\phi,\cC)-x\cB_t(x)
\end{equation}
is called the hedger's {\it gained value} associated with the replicating strategy $(x,p^r_0,\phi,\cC)$.
\end{definition}

Note that the equality $p^g_0(x,\cC)=p^r_0(x,\cC)$ is always satisfied since $V_0(x,p^r_0(x,\cC),\phi,\cC)= x+p^r_0(x,\cC)$.
The financial interpretation of a hedger's replication cost $p^r_0(x,\cC) $ for a given contract $\cC \in \sC $ is fairly straightforward. It represents either an increase or a reduction of the hedger's initial endowment $x$, which is required to implement a trading strategy ensuring that the hedger's wealth at time $T$, after the terminal payoff of the contract has been settled, perfectly matches the value at time $T$ of the original initial endowment $x$ invested in the cash account.

As expected, for the null contract $\cN =(0,0)$, the self-financing strategy $(x,0,\phi, \cN )$, where the portfolio $\phi$ hinges on keeping all money in the bank account $\cB(x)$, is a replicating strategy for $\cN$ such that the gained value satisfies $p^g_t(x, \cN )=0$ for all $t \in [0,T]$. The fact that the trading strategy $(x,0,\phi,\cN)$ is self-financing was postulated in Section \ref{sec3.2} but, obviously, this assumption needs to be verified for each particular market model under study. Note also that the uniqueness of a replication cost $p^r_0(x,\cC)$ is not ensured and, in fact, there is no reason to expect that it will always hold in every market model satisfying either Definition \ref{defarbn} or Definition  \ref{defarbt} (for a counterexample, see Proposition \ref{propnewd}).


Let us make some comments on replication costs. Suppose first that Assumption \ref{ass1.1w} is met. Let us assume that the interval $I_0^r(x,\cC)$ introduced in part (ii) in Lemma \ref{lemfun} is nonempty. We already know that any number $p \in (\ovepl_0(x,\cC), \undpa_0(x,\cC))$ is necessarily a replication cost and a fair price for $\cC$ but it is unclear whether the minimal replication cost is well defined. In principle, it may also happen that there exists a replication cost either equal to, or strictly greater than, $\undpa_0(x,\cC)$.

Let us now examine the case where the interval $I_0^r(x,\cC)$ is empty. If the contract $\cC$ can be replicated, then it may then happen that $\ovepl_0(x,\cC)=\undpa_0(x,\cC)=p^r_0(x,\cC)$. It is not obvious, however, whether $p^r_0(x,\cC)$ would be in that case a hedger's fair price, since it may occur that a strict superhedging strategy with the same initial cost exists. Furthermore, it is also possible that $\undpa_0(x,\cC) < p^r_0(x,\cC)$ meaning that a strict superhedging may be in fact less expensive than replication.
If the contract $\cC$ cannot be replicated, then it is not clear whether $\ovep^f_0(x,\cC)$ is a fair price, although it coincides with the upper bound for loss-generating costs and with the lower bound for strict superhedging costs.

If Assumption \ref{ass1.1s} is met, then one can be a bit more specific. If the set of replication costs is nonempty,
then only the lowest cost of replication can be a fair price (indeed, if $p_1$ and $p_2$ are two replication costs such that $p_1<p_2$,
then $p_2$ is also a strict superhedging cost and thus it is not a fair price). Consequently, if the set of all possible replication costs
for a given contract is bounded from below then either (a) the lower bound for replication costs is not a replication cost and none
of replication costs is a fair price or (b) the lower bound is a replication cost and it is a candidate for a maximal fair
price for the contract.

To conclude, in a nonlinear market model, which is assumed to have the no-arbitrage property with respect to the null contract (or even the no-arbitrage property for the trading desk), a replication cost may fail to be a fair hedger's price. To amend this shortcoming of a general nonlinear setup, we introduce in the next section a particular class of models in which replication yields a unique fair price for any contract $\cC$ belonging to a predetermined family $\sC$ and such that a replicating strategy for $\cC$ exists.

\subsection{Market Regularity on $[0,T]$}          \label{sec3.9}

Once again, we consider the hedger with the initial endowment $x$ at time 0. Intuitively, the concept of {\it regularity} with respect to a given family $\sC$ of contracts is motivated by our desire to ensure that, for any contract from ${\sC}$, the cost of replication is never higher than the minimal cost of superhedging and, in addition, the cost of replication is a fair hedger's price, in the sense of Definition \ref{def:deffi7b}.


\begin{definition} \label{def:de7bb}
We say that the market model $\cM=(\cS,\cD, \cB, \sC , \Psi^{0,x} (\sC ) )$ is {\it regular on $[0,T]$ with respect to} ${\sC}$ if
Assumption \ref{ass1.1w} is met and for every replicable contract $\cC \in {\sC}$ and for any replicating strategy $(x , p^r_0(x),\phi,\cX)$ the following properties hold:
\begin{enumerate}[(i)]\addtolength{\itemsep}{-.5\baselineskip}
\item if $p$ is such that there exists $(x,p,\phi,\cC) \in \Psi^{0,x}(p,\cC)$ satisfying
\begin{equation}  \label{pri1x}
\bP \big( \widetilde{V}_T(x,p ,\phi,\cC)\geq x \big)=1 ,
\end{equation}
then $p \geq p^r_0(x)$;
\item  if $p$ is such that there exists $(x,p,\phi,\cC)\in \Psi^{0,x} (p,\cC)$ such that
\begin{equation}  \label{pri2x}
\bP \big( \widetilde{V}_T(x,p,\phi,\cC)\geq x \big)=1
\end{equation}
and
\begin{equation}  \label{pri3x}
\bP \big( \widetilde{V}_T(x,p,\phi,\cC)>x \big) > 0 ,
\end{equation}
then  $p > p^r_0(x)$.
\end{enumerate}
\end{definition}

By applying Definition \ref{def:de7bb} to the null contract $\cN=(0,0)$, we deduce that any regular market
model is arbitrage-free for the hedger with respect to the null contract. It is not clear, however,
whether an arbitrage opportunity for the trading desk may arise in a regular model.

Condition (i) in Definition \ref{def:de7bb} means that superhedging cannot be less expensive than replication,
whereas condition (ii) postulates that any strict superhedging has a strictly higher cost than replication.
It is important to observe that condition (i) implies that the replication cost $p^r_0(x)$ for $\cC $ is unique.
Moreover, if condition (i) holds, then condition (ii) is equivalent to the following condition
\begin{itemize}
\item[(iii)] if $p$ is such that there exists $(x,p,\phi,\cC)\in \Psi^{0,x} (p,\cC)$ satisfying
\begin{equation}\label{pri2b}
\bP \big( \widetilde{V}_T (x,p,\phi,\cC) \geq x \big)=1,
\end{equation}
then the following implication is valid: if $p=p^r_0(x)$, then
\begin{equation}\label{pri3b}
\bP \big( \widetilde{V}_T (x,p,\phi,\cC)=x \big)=1.
\end{equation}
\end{itemize}

\begin{remark}
In the special case of European claims with maturity $T$ and no trading adjustments, conditions (i) and (ii) correspond to the comparison and strict comparison properties for solutions to BSDEs satisfied by the wealth process with different terminal conditions. In fact, the same idea underpins Definition 2.7 of the nonlinear pricing system introduced by El Karoui and Quenez \cite{EQ1997}. Using similar arguments as El Karoui and Quenez \cite{EQ1997}, we will show that the regularity of a market model can be established for a large variety of financial models using a BSDE approach.
\end{remark}

\subsubsection{Replicable Contracts in Regular Markets}

In this section, we focus on contracts that can be replicated. Proposition  \ref{prop3.13} shows that, in a regular market model, the cost
of replication is the unique fair price of a contract $\cC \in \sC $ that can be replicated and
the equalities $p^r_0(x,\cC) =\undpa_0(x,\cC)=\ovepl_0(x,\cC)$ hold for such a contract.
This means that the replication of a contract is indeed an effective method of valuation within the framework of a regular model,
although this statement is not necessarily true when dealing with an arbitrary nonlinear market model.

\begin{proposition} \label{prop3.13}
 Let a market model $\cM=(\cS,\cD, \cB, \sC , \Psi^{0,x} (\sC ) )$  be regular on $[0,T]$ with respect to ${\sC}$. Then for every contract $\cC \in \sC $ that can be replicated on $[0,T]$ we have:
\begin{enumerate}[(i)]\addtolength{\itemsep}{-.6\baselineskip}
\item the replication cost $p^r_0(x,\cC)$ is unique,
\item $p^r_0(x,\cC)$ is the maximal fair price and the upper bound
for loss-generating  costs, that is, $p^r_0(x,\cC)=\ovep^f_0(x,\cC)= \ovepl_0(x,\cC)$,
\item $p^r_0(x,\cC)$ is the lower bound for superhedging costs and strict superhedging costs,
     that is, $p^r_0(x,\cC)=\undp^s_0(x,\cC)=\undpa_0(x,\cC)$.
\end{enumerate}
\end{proposition}

\begin{proof}
As was mentioned, the uniqueness of the replication cost $p^r_0(x,\cC)$ is an immediate
consequence of condition (i) in  Definition \ref{def:de7bb}. Hence the interval $I_0^r(x,\cC)$ defined in part (ii) in Lemma \ref{lemfun}
is empty and thus, in particular, $p^r_0(x,\cC) \geq \ovep^f_0(x,\cC)$.  Condition (ii) in  Definition \ref{def:de7bb} implies that there is no trading strategy  $(x,p^r_0(x,\cC),\phi,\cC)\in \Psi^{0,x} ( \sC )$ such that conditions
\eqref{pri2x} and \eqref{pri3x} are satisfied. This means that no hedger's arbitrage opportunity arises (that is,
no strict superhedging strategy for $\cC $ exists) if $\cC $ is traded at time 0 at its replication cost  $p^r_0(x,\cC)$
and thus we conclude that $p^r_0(x,\cC)$ is the maximal fair price meaning that $p^r_0(x,\cC)= \ovep^f_0(x,\cC)$.
The remaining equalities are immediate consequences of Lemma \ref{lemfun}.
\end{proof}

\subsubsection{Nonreplicable Contracts}

Let us make some comments on the properties of a contract $\cC$ for which replication in $\cM=(\cS,\cD, \cB, \sC , \Psi^{0,x} (\sC ) )$ is not feasible. If Assumption \ref{ass1.1w} is satisfied, then from Lemma \ref{lemfun} we obtain the following equalities
$\ovepl_0(x,\cC)={\ovep}^f_0(x,\cC)=\undp^s_0(x,\cC)=\undpa_0(x,\cC)$. Obviously, it would be interesting to know whether the common value
is a fair price. Unfortunately, the definitive answer is not available, in general, since it may
happen that  it is indeed the maximal fair price, but it may also occur that it represents the cost of a strict superhedging strategy.
Furthermore, it is not clear at all how to proceed to compute the value of ${\ovep}^f_0(x,\cC)$ (it would be perhaps enough to know this value
since any strictly lower value is a loss-generating cost). Under either Assumption \ref{ass1.1s} or when we postulate that the market model $\cM$ is regular, we obtain the same conclusions as under  Assumption \ref{ass1.1w} and thus they are not helpful when analyzing the valuation of nonreplicable contracts.

\subsubsection{Nonregular Market Model}  \label{sec4.2.3}

Our next goal is to illustrate the issue of regularity of a model by providing a simple, albeit admittedly artificial,
example of a model, which is not regular.  We now assume that the hedger's endowment at time 0 is null and we start by placing ourselves within the framework of Bergman's \cite{B1995} model with differential borrowing and lending interest rates (see also Korn \cite{K1995}, Nie and Rutkowski \cite{NR2015} and Mercurio \cite{M2013}). In particular, the stock price $S^1$ is driven by the Black-Scholes dynamics and constant interest rates satisfy $r^b > r^l$. It is straightforward to verify that Bergman's model $\cM_B =(S^1,B^l,B^b, \sC ,\Psi^{0,0}(\sC ))$ satisfies Definition~\ref{defarbn} if, for instance, we take as $\sC$ the class of long and short positions in all European put options written on the stock $S^1$ and maturing at $T$. Moreover, the hedger is able to replicate, without borrowing any cash, the short position in the European put option on $S^1$ with maturity $T$ and any strike $K>0$. The hedger's price of the European put in Bergman's model is thus given by the classical Black-Scholes formula with the interest rate equal to $r^l$ and we henceforth denoted it as $P_t(K)$ for every $t \in [0,T]$.

In the next step, we fix strike $K>0$ and the date $U \in (0,T)$. For concreteness, we assume that $T-U=1$ and $r^l=0$.
In particular, the class $\sC$ comprise only short and long positions in the put option with strike $K$. We complement the model $\cM_B$ by introducing an additional  risky asset with the following price process with nondecreasing sample paths
$$
S^2_t=\I_{[0,T]}(t)+(K(t-U))(P_U(K))^{-1}\I_{\{2P_U(K)>K\}}\I_{[U,T]}(t).
$$
Since $r^l=0$, we have that $\bP (2P_U(K)>K)>0$,  due to the fact that  $2P_U(K)>K$ for $S_U$ sufficiently close to zero.
It is obvious that $S^2_T=1$ on the event $\{ 2P_U(K) \leq K \}$. On the other hand, on the event $\{ 2P_U(K)>K \}$, we have
$$
1<S^2_T=1+K (P_U(K))^{-1}<3=e^{ \ln 3}.
$$
We henceforth assume that $r^b > \ln 3$ in order to ensure that the interest rate $r^b$ dominates the rate of return on the asset $S^2$.
Obviously, the put option with strike $K$ can be replicated in $\cM=(B^l,B^b,S^1,S^2,\sC,\Psi^{0,0}(\sC ))$
(through the same strategy as in $\cM_B$, that is, using only $B^l$ and $S^1$ for trading) and thus its replication cost in $\cM$ is given by the Black-Scholes price $P(K)$. We are ready to prove the following proposition describing the properties of $\cM$.
It is easy to check that Assumption \ref{ass1.1w}  (as well as Assumption \ref{ass1.1s}) is satisfied by the model $\cM$.

\begin{proposition}  \label{propnewd}
The market model $\cM =(B^l,B^b,S^1,S^2,\sC,\Psi^{0,0}(\sC))$ has the following properties:
\begin{enumerate}[(i)]\addtolength{\itemsep}{-.6\baselineskip}
\item $\cM$ has the no-arbitrage property with respect to the null contract,
\item $\cM$ has the no-arbitrage property for the trading desk,
\item $\cM$ is nonregular and the extended model $\widetilde \cM =(B^l,B^b,S^1,S^2,S^3,\sC,\Psi^{0,0}(\sC))$ where
$S^3=P(K)$ does not have the no-arbitrage property with respect to the null contract.
\item replication cost for the put option is not unique and the minimal cost of replication is given by the solution $Y$
to the BSDE
\begin{equation} \label{bsdesput}
dY_t=\sum_{i=1}^2 \xi_t^i\,dS^i_t , \quad Y_T=(K-S^1_T)^+ .
\end{equation}
\end{enumerate}
\end{proposition}

\begin{proof}
Assertion (i) is easy to check. Intuitively, the increasing process $S^2$ can be seen as an alternative (artificial) lending account to the constant lending account $B^l=1$ and thus, if the hedger has a surplus of cash, then he will invest it in the asset $S^2$, rather than in the lending account $B^l$. It is thus enough to note that the model $(S^1, \bar{B}^l=S^2,B^b, \sC ,\Psi^{0,0}(\sC ))$ satisfies Definition~\ref{defarbn}, as it can be seen as another instance of Bergman's setup with a
random lending rate.

For part (ii),  we focus on the model $(S^1,\bar{B}^l=S^2,B^b,\sC,\Psi^{0,0}(\sC ))$ and we observe that equation \eqref{eq:dynwtVcom}
reduces to
\[
d\Vcomt_t(x_1,-x_1,\varphi, \bar \varphi, A,0,0)= (\xi_t^1 +\bar \xi_t^1)\,d \wt S^{1}_t+(\psi_t^{0,b}+\bar  \psi_t^{0,b} )\, dB_t^{0,b,l},
\]
where $A$ represents the put option, $\wt S^1=S^1 (\bar{B}^l)^{-1}$ and $B^{0,b,l}=B^b (\bar{B}^l)^{-1}$.  Since the process $\wt S^1$ admits a martingale measure, the process $B^{0,b,l}$ is non-increasing and, by assumption, the processes $\psi^{0,b}$ and $\bar \psi^{0,b}$ are nonnegative, it is clear that no arbitrage opportunity for the trading desk may arise in the considered model, and thus also in $\cM$.

To prove (iii), we will show that the hedger who sells the put option at time 0 at its Black-Scholes price $P_0(K)$ can construct
an arbitrage opportunity. To see this, assume that the hedger uses the replicating strategy for the put on the interval $[0,U]$.
If the event $\{ 2P_U(K)  \leq K \}$ occurs, then he continues to replicate the put until its maturity date $T$.
In contrast, if the event $\{ 2P_U(K)>K \}$ occurs, then he buys $P_U(K)/S^2_U=P_U(K)$ of shares of the asset $S^2$ and holds
them till $T$. Then the hedger's wealth at $T$, after delivery of the cash amount $(K-S^1_T)^+$ to the buyer
of the option, satisfies
\begin{align*}
 V_T (0,P_0(K),\phi,-P(K)) &= \left( \frac{P_U(K)}{S^2_U}\, S^2_T - (K-S^1_T)^+ \right) \I_{\{ 2P_U(K)> K \}}+{0\cdot \I_{\{ 2P_U(K)\leq K \}}} \\ &=\left( P_U(K) (1+K (P_U(K))^{-1}) - (K-S^1_T)^+ \right) \I_{\{ 2P_U(K)> K \}} \\ &> \left( 1.5 K - (K-S^1_T)^+ \right) \I_{\{ 2P_U(K)> K \}} > 0.5 K \I_{\{ 2P_U(K)> K \}}.
\end{align*}
Since $\bP ( 2P_U(K)> K )>0$ and it is clear that the wealth is always non-negative (so that the strategy
is admissible), the strategy described above is an arbitrage opportunity for the hedger.
Also, a strict superhedging strategy for the claim $P_T(K)$ can be obtained using the initial wealth equal to the replication cost. Indeed, using the premium $P_0(K)$, the hedger may either replicate $\zeta_1 := P_T(K)$ or strictly superhedge $P_T(K)$ by producing the terminal payoff
$$
\zeta_2=P_T(K) \I_{\{ 2P_U(K)\leq  K \}}+P_U(K) (1+K (P_U(K))^{-1}) \I_{\{ 2P_U(K)> K \}}.
$$
It is easy to check that $\zeta_2 \geq \zeta_1 $ and $\bP ( \zeta_2 > \zeta_1 ) >0 $.
We thus conclude that the model $\cM$ is nonregular and the extended model $\widetilde \cM $  does not satisfy Definition~\ref{defarbn}.

For the last assertion, we note that the put option can be replicated in the model $(S^1,\bar{B}^l=S^2,\sC,\Psi^{0,0}(\sC ))$,
which is an extension of the Black-Scholes model in which the interest rate is random. This claim follows from the existence
and uniqueness of a solution $(Y,\phi^1,\phi^2)$ to BSDE \eqref{bsdesput} with Lipschitz continuous coefficients and the bounded
terminal condition. Furthermore, the component $\phi^2$ of the replicating strategy is nonnegative and thus the same strategy replicates
the option in $\cM$ and it is the least expensive replicating strategy in $\cM$.
\end{proof}

\subsection{Replication and Market Regularity on $[t,T]$} \label{sec4.1}

The notion of the hedger's gained value $p^g_t (x,\cC), \, t \in [0,T)$ reduces to the classical no-arbitrage price obtained through replication in the linear setup  provided that the only cash flow of $A$ after time 0 is the terminal payoff, which equals $A_T-A_{T-}$.
Unfortunately, in a general nonlinear setup considered in this work, the financial interpretation of the hedger's gained
value at time $t>0$ is less transparent, since it depends on the hedger's initial endowment, the past cash flows of a contract
and the strategy implemented by the hedger on $[0,t]$. The following definition mimics Definition \ref{defprice1},
but focuses on the restriction of a contract $\cC $ to the interval $[t,T]$. Note that here the discounted wealth process is given
by equation \eqref{edd}. It is assumed in this section that $\cC^t$ can be replicated on $[t,T)$
at some initial cost $p^r_t$ at time $t$, in the sense of the following definition.

\begin{definition} \label{defprice2}
For a fixed $t \in [0,T]$, let $p^r_t$ be a ${\cal G}_{t}$-measurable random variable. If there exists
a trading strategy $(x_t, p^r_t,\phi^t,\cC^t) \in \Psi^{t,x_t} ( \sC)$ such that
\begin{equation} \label{xprice2}
\wt{V}_T(x_t, p^r_t,\phi^t,\cC^t)=x_t ,
\end{equation}
then $p^r_t=p^r_t (x_t,\cC^t)$ is called a {\it replication cost} at time $t$ for the contract $\cC^t$ relative to
the hedger's endowment $x_t$ at time $t$.
\end{definition}

As expected, Definition \ref{def:de7bb}, and thus also Proposition \ref{prop3.13}, can be extended to any date $t \in [0,T)$.

\begin{definition} \label{def:de7bc}
We say that a market model $\cM=(\cS,\cD, \cB, \sC ,\Psi^{t,x_t} (\sC ) )$ is
{\it regular on $[t,T]$ with respect to} ${\sC}$ if Assumption \ref{ass1.1w} is met and the following
properties hold for every contract $\cC \in {\sC}$ that can be replicated:
\begin{enumerate}[(i)]\addtolength{\itemsep}{-.5\baselineskip}
\item if $p_t \in \cG_t$ and there exists $(x_t,p_t,\phi^t,\cC^t) \in \Psi^{t,x_t} (\sC)$ satisfying
\begin{equation}\label{pri1y}
\bP \big( \widetilde{V}_T(x_t,p_t,\phi^t,\cC^t) \geq x_t\big)=1,
\end{equation}
then $p_t \geq p^r_t (x_t,\cC^t)$;
\item if $p_t  \in \cG_t $ and there exists  $(x_t,p_t,\phi^t,\cC^t) \in \Psi^{t,x_t} (\sC)$ such that
for some $D \in \cG_t$
\begin{equation}\label{pri2y}
\bP \big( \I_D \widetilde{V}_T(x_t,p_t,\phi^t,\cC^t) \geq \I_D x_t\big)=1
\end{equation}
and
\begin{equation}\label{pri3y}
\bP \big( \I_D \widetilde{V}_T(x_t,p_t,\phi^t,\cC^t) > \I_D x_t\big) > 0,
\end{equation}
then  $ \bP ( \I_D\, p_t > \I_D p^r_t (x_t,\cC^t) ) >0$.
\end{enumerate}
\end{definition}

Similarly as for $t=0$, if condition (i) holds, then condition (ii) is equivalent to the following condition:
\begin{itemize}
\item[(iii)] if $p_t \in \cG_t $ and there exists  $(x_t,p_t,\phi^t,\cC^t) \in \Psi^{t,x_t} (\sC)$ such
that for some event $D \in \cG_t$
\begin{equation}\label{pri2}
\bP \big( \I_D \widetilde{V}_T(x_t,p_t,\phi^t,\cC^t) \geq \I_D x_t\big)=1 ,
\end{equation}
then the following implication holds: if $\I_D p_t=\I_D p^r_t (x_t,\cC^t)$, then
\begin{equation}\label{pri3}
\bP \big( \I_D \widetilde{V}_T(x_t,p_t,\phi^t,\cC^t)=\I_D x_t\big) =1 .
\end{equation}
\end{itemize}

In the following extension of Proposition \ref{prop3.13}, we assume that the hedger's endowment $x_t$ at time $t$
is given and the contract $\cC$ has all cash flows on $(t,T]$ so that $\cC^t=\cC$. We now search for the hedger's fair price for $\cC$ at time $t$ assuming that a replication strategy exists. A closely related, but not identical, valuation problem is
studied in Section \ref{sec4.2} where we study the valuation at time $t$ of a contract originated at time 0.

\begin{proposition} \label{prop3.13n}
 Let a market model $\cM=(\cS,\cD, \cB, \sC , \Psi^{t,x_t} (\sC ) )$ be regular on $[t,T]$ with respect to the class ${\sC}$. Then for every contract $\cC \in \sC $ such that $\cC=\cC^t$ and which can be replicated on $[t,T]$, we have:
\begin{enumerate}[(i)]\addtolength{\itemsep}{-.6\baselineskip}
\item the replication cost $p^r_t(x_t,\cC)$ is unique,
\item $p^r_t(x_t,\cC)$ is the maximal fair price and the upper bound
for loss-generating costs, that is, $p^r_t(x_t,\cC)=\ovep^f_t(x_t,\cC)= \ovepl_t(x_t,\cC)$,
\item $p^r_t(x_t ,\cC)$ is the lower bound for superhedging costs and strict superhedging costs,
     that is, $p^r_t (x_t ,\cC)=\undp^s_t(x_t,\cC) =\undpa_t(x_t,\cC)$.
\end{enumerate}
\end{proposition}

The proof of Proposition \ref{prop3.13n} is very similar to the proof of Proposition \ref{prop3.13} and thus it is omitted.
According to Proposition \ref{prop3.13n}, in any market model regular on $[t,T]$, a replication cost is unique and it is the
maximal fair price at time $t$ for the hedger with the endowment $x_t$ at time $t$.

\section{Pricing by Replication in Regular Markets} \label{sec4}

In this section, it is assumed that a market model $\cM$ is regular. Our goal is to examine the properties of various kinds of prices for a contract $\cC $ under the assumption that it can be replicated on $[t,T]$ for every $t \in [0,T)$. Recall that for any fixed $t\in [0,T)$ we denote by $(x_t, p_t,\phi^t,\cC^t)$ a hedger's trading strategy starting at time $t$ with a $\cG_t$-measurable endowment $x_t$ when a contract $\cC^t$ is traded at a $\cG_t$-measurable price $p_t$. For simplicity, we focus on contracts $\cC=(A,\cX)$ with a constant maturity date $T$, which correspond to non-defaultable contracts of European style. To deal with the counterparty credit risk, it suffices to replace $A$ with the process $A^\sharp $ introduced in Section \ref{sec2.8.1} and a fixed maturity $T$ with the effective maturity of a contract at hand (for instance, by $T \wedge \tau $ where $\tau $ is the random time of the first default or, more generally, by the effective settlement date of a contract in the presence of the gap risk). Furthermore, in the case of contracts of an American style or game options, the effective settlement date is also affected by respective decisions of both parties to prematurely terminate the contract.

\subsection{Hedger's Ex-dividend Price at Time $t$} \label{sec4.2}

Since it was postulated that the initial endowment of the hedger at time 0 equals $x$, it is clear that any application of
Definition \ref{defprice2} should be complemented by a financial interpretation of the hedger's
endowment $x_t$ at time~$t$ and a relationship between the quantity $x_t$ and the hedger's initial endowment $x$ should be clarified.
One may consider several alternative specifications for $x_t$, which correspond to different financial interpretations
of valuation problems under study:
\begin{enumerate}
\item A first natural choice is to set $x_t=x_t(x):=x\cB_t(x)$ meaning that the hedger has not been dynamically hedging the contract between time 0 and time $t$ (this particular convention was adopted in Bielecki and Rutkowski \cite{BR2015} and Nie and Rutkowski \cite{NR2015,NR2016a}). Then the quantity $p^r_t(x_t,\cC^t)$ is the future fair price at time $t$ of the contract $\cC^t$, as seen at time 0 by the hedger with the endowment $x$ at time 0, who decided to postpone trading in $\cC$ to time $t$. This specification of $x_t$ could be convenient if one wishes to study, for instance, the issue of valuation at time 0 of the option with the expiration date $t$ written on the contract $\cC^t$.
\item Alternatively, one may postulate that the contract was entered into by the hedger at time 0 at his replication cost $p^r_0(x,\cC)$ and he decided to keep his position unhedged. In that case, the initial price, the cash flows, and the adjustments should be appropriately accounted for when computing the actual hedger's endowment $x_t$ at time~$t$ using a particular market model.
\item Next, one may assume that the contract was entered into by the hedger at time 0 at the price $p^r_0(x,\cC)$ and
was hedged by him on $[0,t]$ through a replicating strategy $ \phi $, as given by Definition \ref{defprice1}. Then the hedger's endowment at time $t>0$ equals $x_t=V_t(x,p^r_0(x,\cC),\phi,\cC)$ and it is natural to expect that the equality $p^r_t(x_t,\cC^t)=0$ will hold for all $t \in (0,T]$.
\item Finally, one can simply postulate that the hedger's endowment $x_t$ at time $t$ is exogenously specified.
Then Definition \ref{defprice2} reduces in fact to Definition \ref{defprice1} with essentially identical financial
interpretation: we define the hedger's initial price at time $t$ for the contract $\cC^t$  given
his initial endowment $x_t$ at time $t$. Of course, under this convention no relationship between
the quantities $x$ and $x_t$ exists.
\end{enumerate}

In the next definition, we apply Definition \ref{defprice2} to the first of the above-mentioned specifications
of $x_t$, that is, we set $x_t=x_t(x) := x \cB_t(x)$. As in Definition \ref{defprice2},  the discounted wealth is given by \eqref{edd}.

\begin{definition} \label{defprice3} {\rm
For a fixed $t \in [0,T]$, assume that $p^e_t$ is a ${\cal G}_{t}$-measurable random variable.
If there exists a trading strategy $(x_t(x) , p^e_t,\phi^t,\cC^t) \in \Psi^{t, x_t(x)}(\sC)$ such that
\begin{equation} \label{xprice3}
\wt{V}_T(x_t(x),p^e_t,\phi^t,\cC^t)=x_t(x),
\end{equation}
then $p^e_t=p^e_t (x,\cC^t)$ is called the {\it hedger's ex-dividend price} at time $t$ for the contract $\cC^t $.}
\end{definition}

Note that $p^r_0(x,\cC)=p^g_0(x,\cC)=p^e_0(x,\cC)$ and $p^g_T (x,\cC)=p^e_T(x,\cC^T) =0$. The price given in Definition \ref{defprice3} is suitable when dealing with derivatives written on the contract $\cC^t$ as an underlying asset or, simply, when the hedger would like to compute the future fair price for $\cC^t$ without actually entering into the contract at time 0. Furthermore, it can also be used to define a proxy for the exit price of the contract $\cC^t $.

It is natural to ask whether the processes $p^g_t(x,\cC)$ and $p^e_t(x,\cC^t)$ coincide for all $t \in [0,T]$. We will argue that the equality $p^g_t(x,\cC)=p^e_t(x,\cC^t)$ is valid for every $t$ when the valuation problem is local, but it is not necessarily true for a global valuation problem (see Proposition \ref{prop5.2}). The reason is that in the former case the two processes satisfy identical BSDE, whereas in the latter case one obtains a generalized BSDE for the former process and a classical BSDE for the latter. It is also intuitively clear that in the case of the global valuation problem the two processes will typically differ, since the value of $p^e_t(x,\cC^t)$ is clearly independent of the hedger's trading strategy on $[0,t]$, as opposed to $p^g_t(x,\cC)$, which may depend on the whole history of his trading. Since most valuation problems encountered in the existing literature
have a local nature, to the best of our knowledge, this particular issue was not yet examined by other authors.

\subsection{Exit Price}           \label{sec4.3}

The issue of valuation at time $t$ is also important when we ask the following question: at which {\it exit price} a contract entered into by the hedger at time 0 can be unwound by him at time $t$.  In theory, the easiest way to unwind at time $t$ a contract originated at time 0 at the price $p^r_0(x,\cC)=p^g_0(x,\cC)$ would be to transfer all obligations associated with the remaining part of the contract on $[t,T]$ to another trader. Accordingly, the gained value $p^g_t (x,\cC)$ for $0< t <T$ would be the amount of cash, which the hedger would be willing to pay to another trader who would then take the hedger's position from time $t$ onwards. This argument leads to the following definition of the hedger's exit price.

\begin{definition} \label{defprice4}
The \textit{exit price} for the contract $\cC$ entered into at time 0 by the hedger with the initial endowment $x$ is given by the equality $p^m_t(x,\cC) := - p^g_t(x,\cC)$ for every $t \in [0,T]$.
\end{definition}

Definition \ref{defprice4} reflects the market practice where the exit price is related to the concept of unwinding the existing
contract at time $t$ at its current market value. Note, in particular, that the equality $p^g_t(x,\cC)+p^m_t(x,\cC)= 0 $
holds for every $t \in [0,T]$, meaning that the net value of the fully hedged position is null at any moment when the contract is
marked to market. Unfortunately, a practical implementation of Definition \ref{defprice4} could prove difficult, especially when dealing with a global valuation problem, since it would require to keep track of past cash flows from the contract and gains from the hedging strategy (of course, provided the hedging strategy was implemented by the hedger). We thus contend that the proxy for the exit value $p^m_t (x,\cC) := - p^e_t(x,\cC)$ could be more suitable for most practical purposes when facing a global valuation problem. Since the equality $p^g_t(x,\cC)=p^e_t(x,\cC)$ holds when the valuation problem is local, the issue of the choice of a marking to market convention is clearly immaterial in that case.

\subsection{Offsetting Price}             \label{sec4.4}

Assume that the hedger is unable to transfer to another trader his existing position in the contract $\cC $ entered into at time 0. Then he may attempt to offset his future obligations associated with $\cC$ by taking the opposite position in an ``equivalent'' contract. In the next definition, we postulate that the hedger attempts to unwind his long position in $\cC=(A,\cX)$ at time $t$ by entering into an \textit{offsetting contract} $(-A^t,\cY^t)$. It is also assumed here that he liquidates at time $t$ the replicating portfolio for $\cC $ so that his endowment at time $t$ equals $\wh V_t(x,\cC) := V_t(x,p^r_0(x,\cC), \wh{\phi} ,\cC)$, where $\wh{\phi}$ is a replicating strategy for $\cC$ on $[0,T]$.

\begin{definition} \label{defprice5} {\rm
For a fixed $t \in [0,T]$, let $p^o_t$ be a ${\cal G}_{t}$-measurable random variable.
 If there exists an admissible trading strategy $(\wh V_t(x,\cC),p^o_t,\phi^t,0,\cX^t+\cY^t)$
on $[t,T]$ such that
\begin{equation} \label{xprice5}
\wt{V}_T( \wh V_t(x,\cC),p^o_t,\phi^t,0,\cX^t+\cY^t)=x,
\end{equation}
then $p^o_t= p^o_t(x,\cC^t) $ is called the {\it offsetting price} of $\cC^t=(A^t ,\cX^t )$ through $(-A^t ,\cY^t )$ at time $t$.}
\end{definition}

Definition  \ref{defprice5} takes into account the fact that the cash flows of $A^t$ and $-A^t$
(and perhaps also some cash flows associated with the corresponding adjustments $\cX^t$ and $\cY^t$) offset one another and thus only the residual cash flows need to be accounted for when computing the price at which the contract $\cC $ can be unwound by the hedger at time $t$.

In the special case where the equality $\cX^t+\cY^t =0$ holds for all $t \in [0,T]$ (that is, the offsetting is perfect)
we obtain the equality $p^o_t (x,\cC^t)=- p^g_t(x,\cC)$ since then, in view of \eqref{xprice1}, we have that
$$
 \wh V_t (x,\cC)+  p^o_t(x,\cC^t)=p^g_t (x,\cC)+ x \cB_t(x)+ p^o_t(x,\cC^t)=x \cB_t(x),
$$
where $x \cB_t(x)$ is the cash amount that is required to replicate the null contract $(0,0)$ on $[t,T]$.

\section{A BSDE Approach to Nonlinear Pricing}      \label{sec6}

For each definition of the price, one may attempt to derive the corresponding backward stochastic differential equation (BSDE) by combining their definitions with dynamics \eqref{wealth3-1} of the hedger's wealth or, even more conveniently, with dynamics \eqref{wealthdyna} of his discounted wealth. Subsequently, each particular valuation problem can be addressed by solving a suitable BSDE. In addition, one may use a BSDE approach to establish the regularity property of a market model at hand. To this end, one may either use the existing (strict) comparison theorems for solutions to BSDEs or, if needed, to establish original results. In this work, we are not analyzing these issues in detail since our goal is
merely to derive BSDEs for the gained value and the ex-dividend price in a particular setup with trading adjustments and to emphasize the
difference between local and global pricing problems. We conclude the paper by outlining important issues related to a BSDE approach to the counterparty credit risk.

\subsection{BSDE for the Gained Value} \label{sec6.1}

We will first derive a generic BSDE associated with the hedger's gained value $p^g(x,\cC)$ introduced in Definition \ref{defprice1}. For concreteness, we focus on the case of $x \geq 0$ so that the equality $\cB(x)=B^l$ is valid.  To simplify the presentation, we also postulate that $B^{i,l}=B^{i,b}=B^i$ for $i=1,2,\ldots ,d$ and we consider trading strategies satisfying
the funding constraint $\wh \xi^i_t S^i_t+\wh \psi^{i}_t B^{i}_t=0$ for all $i=1,2, \dots , d$ and every $t \in [0,T]$.
Of course, the case where $x<0$ can be dealt with in an analogous manner.
Recall that (see \eqref{sicld})
\[
\wt S^{i,cld}_t(x) := (\cB_t(x))^{-1}S_t^i+\int_0^t (\cB_u(x))^{-1}\, dD^i_u.
\]
It is worth recalling that $p^r_0(x,\cC)=p^g_0(x,\cC)$ (see Definition \ref{defprice1}).

\begin{lemma} \label{lemm5.1}
Assume that a trading strategy $(x,p^r_0, \wh \varphi,\cC) \in \Psi^{0,x} ( \sC )$ replicates a contract $\cC$.
Then the processes $\wh Y := \wt{V}^l (x, p^r_0, \wh \phi,\cC) :=(B^{0,l})^{-1} V(x, p^r_0, \wh \phi,\cC)$ and
$\wh Z^i  := \wt B^{i,l} \wh \xi^i$ satisfy the BSDE
\begin{equation}\label{zdynam}
\begin{aligned}
d\wh Y_t=&  \sum_{i=1}^{d} \wh Z_t^i\,d\wh S^{i,cld}_t  - (B^{0,b}_t)^{-1} \bigg( \wh Y_t B^{0,l}_t+\sum_{k=1}^n \alpha^k_t X^k_t \bigg)^- dB_t^{0,b,l}  \\ &+(B^{0,l}_t)^{-1}\, dA_t - \sum_{k=1}^{n} \wh X_t^k\,d\wt \beta_t^{k,l}+\sum_{k=1}^{n} (1-\alpha_t^k) X^k_t\,d (B^{0,l}_t)^{-1}
\end{aligned}
\end{equation}
with the terminal condition $\wh Y_T=x$.
\end{lemma}

\begin{proof}
Under the present assumptions,  \eqref{wealthdyna} and \eqref{newdd} imply
\begin{align} \label{xdynam}
 d\wt V^l_t(x,p^r_0, \wh \varphi,\cC)=&\sum_{i=1}^{d} \widehat \xi_t^i \wt B_t^{i,l}\,d \widehat S_t^{i,cld}+\widehat \psi_t^{0,b}\,dB_t^{0,b,l}+ (B^{0,l}_t)^{-1}\, dA_t - \sum_{k=1}^{n} \wh X_t^k d \wt \beta_t^{k,l} \nonumber \\ &+\sum_{k=1}^{n} (1-\alpha_t^k ) X^k_t\,d(B^{0,l}_t)^{-1},
\end{align}
where $\wt B^{i,l} := (B^{0,l})^{-1}B^{i},\, B^{0,b,l} := (B^{0,l})^{-1}B^{0,b},\, \wt \beta^{k,l} := (B^{0,l})^{-1} \beta^k $.
Equation \eqref{wealth2} and conditions $\wh{\psi}^{0,l} \geq 0, \wh{\psi}^{0,b} \leq 0$ and $\wh{\psi}^{0,l} \wh{\psi}^{0,b}=0$ yield
\begin{equation} \label{xwealth2}
 \wh \psi^{0,l}_t=(B^{0,l}_t)^{-1} \bigg( V_t(x,p^r_0,\wh \phi,\cC)+\sum_{k=1}^n \alpha^k_t X^k_t \bigg)^+
\end{equation}
and
\begin{equation}  \label{xwealth3}
\wh \psi^{0,b}_t=-(B^{0,b}_t)^{-1} \bigg(V_t(x,p^r_0,\wh \phi,\cC)+\sum_{k=1}^n \alpha^k_t X^k_t \bigg)^-.
\end{equation}
Note that the process $\wh \psi^{0,l}$ does not appear in \eqref{xdynam} and the process $\wh \psi^{0,b}$ can
also be eliminated from \eqref{xdynam} by using \eqref{xwealth3}. If we set $\wh Y := \wt{V}^l (x, p^r_0, \wh \phi,\cC)$,  then \eqref{xdynam} can be represented as the BSDE
\begin{equation} \label{ydynam}
\begin{aligned}
d\wh Y_t=&\sum_{i=1}^{d}\widehat \xi_t^i \wt B_t^{i,l}\,d \widehat S_t^{i,cld}- (B^{0,b}_t)^{-1} \bigg( \wh Y_t B^{0,l}_t - \sum_{i=1}^{d}\wh \xi^i_t S^i_t-\sum_{i=1}^{d} \wh \psi^{i}_t B^{i}_t+\sum_{k=1}^n \alpha^k_t X^k_t \bigg)^- dB_t^{0,b,l} \\
&+(B^{0,l}_t)^{-1}\,dA_t-\sum_{k=1}^{n}\wh X_t^k\,d\wt \beta_t^{k,l}+\sum_{k=1}^{n}(1-\alpha_t^k) X^k_t\,d(B^{0,l}_t)^{-1}
\end{aligned}
\end{equation}
with the terminal condition $\wh Y_T=\wt{V}^l_T(x, p^r_0, \wh \phi,\cC)=x$.
In view of equality \eqref{newdd}, BSDE \eqref{ydynam} further simplifies to \eqref{zdynam}.
\end{proof}

In the next result, we focus on a market model satisfying regularity conditions introduced in Definition \ref{def:de7bc}.
From the regularity of a model, it follows that the hedger's gained value $p^g_t(x,\cC)$ is unique for each fixed $t \in [0,T]$.
However, this does not suffice to define the process $p^g(x,\cC)$ and thus in the next result we will make assumptions
regarding BSDE \eqref{zdynam}. First, we postulate that for a given $x \geq 0$ and any contract $\cC \in \sC$ there
exists a unique solution $(\wh Y, \wh Z)$ to \eqref{zdynam} in a suitable space of stochastic processes. Second,
we assume that BSDE \eqref{zdynam} enjoys the following variant of the strict comparison property.

\begin{definition}
The {\it strict comparison} property holds for the BSDE \eqref{zdynam} if  for any contract $\cC \in \sC$ if $(\wh Y^1 , \wh Z^1)$ and $(\wh Y^2 , \wh Z^2)$ are solutions with $\cG_T$-measurable terminal conditions $\xi^1_T \geq \xi^2_T$, respectively, then the equality $\I_D Y^1_t=\I_D Y^2_t$ for some $t \in [0,T)$ and some $D \in \cG_t$ implies that $\I_D \xi^1_T=\I_D \xi^2_T$.
\end{definition}

It is also important to note that one needs to examine the manner in which the inputs in BSDE \eqref{zdynam} (that is, the stochastic processes introduced in Assumption \ref{assumpassets}) may possibly depend on the unknown processes $\wh Y$ and $\wh Z$.

According to Definition \ref{locglob}, the problem examined in this section will be an example of a local valuation problem if we
postulate that $X^k$ and $\beta^k$ satisfy $X^k_t=v^k(t,\wh Y_t ,\wh Z_t)$ and $d\beta_t^{k}= w^{k}(t,\wh Y_t ,\wh Z_t)\, dt$ for some $\mathbb{G}$-progressively measurable mappings $v^k,w^k : \Omega \times [0,T] \times \bR^{d+1}\to \bR $ for every $k=1,2,\dots ,n$. The same valuation problem becomes a global one if $X^k_t=\bar{v}^k(t,\wh Y_{\cdot } ,\wh Z_{\cdot})$ and $d\beta_t^{k}= \bar{w}^k(t,\wh Y_{\cdot} ,\wh Z_{\cdot})\, dt$ for some $\mathbb{G}$-non-anticipative functionals $\bar{v}^k,\bar{w}^k : \Omega \times [0,T] \times \cD ([0,T], \bR^{d+1}) \to \bR $ for every $k=1,2,\dots ,n$, where $\cD ([0,T],\bR^{d+1})$ is the space of $\bR^{d+1}$-valued, $\mathbb{G}$-adapted, c\`adl\`ag processes on $[0,T]$.

From Lemma \ref{lemm5.1}, we deduce that a local valuation problem can be formulated it terms of a classical BSDE. In contrast,
the situation where the inputs depend on the past history of the processes is harder to address, since a global valuation problem
requires to study a generalized BSDE with non-anticipative functionals. Since both situations are covered by Proposition \ref{prop5.1}, we refer the reader to Cheridito and Nam \cite{CN2015} and Zheng and Zong \cite{ZZ2017} for the existence and uniqueness results for generalized BSDEs.
It is worth noting that, to the best of our knowledge, so far no results on the strict comparison property for generalized BSDE are available.
This should be contrasted with the theory of classical BSDEs where the strict comparison theorem plays an important role.
In the next result, we suppose that the dynamics of the wealth process given by \eqref{xdynam} and \eqref{xwealth3} are such that
Assumption \ref{ass1.1w} is satisfied.

\begin{proposition} \label{prop5.1}
Assume that the BSDE \eqref{zdynam} has a unique solution $(\wh{Y},\wh{Z})$ for any  contract
$\cC \in \sC$ and the strict comparison property for solutions to \eqref{zdynam} holds. Then the following assertions are valid:
\begin{enumerate}[(i)]\addtolength{\itemsep}{-.6\baselineskip}
\item the market model is regular on $[t,T]$ for every $t \in [0,T]$;
\item the hedger's gained value satisfies $p^g(x,\cC)= B^{0,l}(\wh{Y} -x)$ where $(\wh{Y},\wh{Z})$ is a solution to BSDE \eqref{zdynam} with the terminal condition $\wh Y_T=x$;
\item the unique replicating strategy $\wh \phi $ for $\cC $ satisfies $\wh \xi^i =(\wt B^{i,l})^{-1}\wh Z^i$ and
the cash components $\wh \psi^{0,l}$ and $\psi^{0,b}$ are given by \eqref{xwealth2} and \eqref{xwealth3}, respectively,
with $V(x,p^r_0,\wh \phi,\cC)$ replaced by $B^{0,l}\wh Y$.
\end{enumerate}
\end{proposition}

\begin{proof}
In view of Definition \ref{def:de7bc}, it is clear that the existence, uniqueness and the strict comparison property
for the BSDE \eqref{zdynam} imply that the market model is regular on $[t,T]$ for every $t \in [0,T]$.
To establish (ii), we recall that the regularity of a model implies that the hedger's gained value $p^g_t(x,\cC)$
is unique. Moreover, we also know that $p^g(x,\cC)$ satisfies for every $t \in [0,T]$  (see \eqref{xprice1})
$$
p^g_t(x,\cC)=V_t(x,p^g_0(x,\cC),\wh \phi,\cC)-x\cB_t(x)=B^{0,l}_t\wh{Y}_t-x\cB_t(x)=B^{0,l}_t(\wh{Y}_t-x),
$$
which establishes the asserted equality $ p^g(x,\cC)=B^{0,l} (\wh Y-x )$. In particular, the hedger's replication cost $p^r_0(x)$ satisfies $p^r_0(x)=\wh Y_0-x$ for any fixed initial endowment $x \geq 0$. Finally, part (iii) is an immediate consequence of Lemma~\ref{lemm5.1}.
\end{proof}

Of course, one needs to check for which models the assumptions of Proposition \ref{prop5.1} are satisfied.
For general results regarding BSDEs driven by one- or multi-dimensional continuous martingales, the reader is referred to
Carbone et al. \cite{CFS2008}, El Karoui and Huang \cite{EH1997} and Nie and Rutkowski \cite{NR2016b}
and the references therein. Typically, a suitable variant of the Lipschitz continuity of a generator to a BSDE is sufficient to
guarantee the desired properties of its solutions. Several instances of nonlinear market models with BSDEs satisfying
the comparison property were studied by Nie and Rutkowski \cite{NR2015,NR2016a,NR2017}, although the concept of a regular model was not formally stated therein. In particular, they analyzed contracts with an endogenous collateral, meaning that an adjustment process
$X^k$ explicitly depends on a solution $\wh Y$ (or even on solutions to the valuation problems for the hedger and the counterparty).

Let us finally mention that since the model examined in this section is a special case of the model studied in Sections \ref{sec3.5} and \ref{sec3.6}, it follows from  Proposition \ref{propoarbi1} that to ensure that the model is arbitrage-free for the trading desk, it suffices to assume that there exists a probability measure $\bQ$, which is equivalent to $\bP$ on $(\Omega , \cG_T)$ and such that the processes $\widehat{S}^{i,\textrm{cld}},\, i=1,2,\ldots, d$ given by \eqref{pri2} are $\bQ$-local martingales.
This assumption is also convenient if one wishes to prove the existence and uniqueness result for BSDE \eqref{zdynam}.

\subsection{BSDE for the Ex-dividend Price}           \label{sec6.2}

Our next goal is to derive the BSDE for the ex-dividend price $p^e(x,\cC)$ introduced in Definition~\ref{defprice3}.
 As in Section \ref{sec6.1}, we work under the assumption that $x \geq 0$. Recall that, for a fixed $t$, the hedger's ex-dividend price is implicitly given by the equality $\wt {V}^l_T(x_t(x),p^e_t,\phi^t,\cC^t)=x_t(x)$
where $x_t(x)=x B^l_t$ and the discounting is done using the process $\cB^t_{\cdot}(x_t(x))$ given by  \eqref{eq:discFactc}.
We henceforth assume that the valuation problem is local. This assumption is essential for validity of Lemma \ref{lemm5.2} and Proposition \ref{prop5.2}, so it cannot be relaxed.

\begin{lemma} \label{lemm5.2}
Assume that a trading strategy $(x_t(x),p^e_t,\varphi^t,\cC^t) \in \Psi^{t,x_t(x)}( \sC )$ replicates a contract $\cC^t$
on $[t,T]$. Then the processes $\bar Y_u := \wt {V}_u(x_t(x),p^e_t,\phi^t,A^t,\mathcal{X}^t)$ and
$\bar Z^i_u:=\wt B^{i,l}_u(\xi^t_u)^i$, $ u\in [t,T]$, satisfy the following BSDE, for all $u \in [t,T]$
\begin{equation}\label{zdyname}
\begin{aligned}
d\bar Y_u=&\sum_{i=1}^{d} \bar Z_u^i\,d\wh S^{i,cld}_u  - (B^{0,b}_u)^{-1} \bigg( \bar Y_u B^{0,l}_u+\sum_{k=1}^n \alpha^k_u X^k_u \bigg)^- dB_u^{0,b,l}\\ &+(B^{0,l}_u)^{-1}\, dA_u - \sum_{k=1}^{n} \wh X_u^k\,d\wt \beta_u^{k,l}+\sum_{k=1}^{n} (1-\alpha_u^k) X^k_u\,d (B^{0,l}_u)^{-1}
\end{aligned}
\end{equation}
with the terminal condition $\bar Y_T=x$.
\end{lemma}

\begin{proof}
Arguing as in the proof of Lemma \ref{lemm5.1}, we conclude that the dynamics of the discounted wealth
$\wt {V}_u(x_t(x),p^e_t,\phi^t,A^t,\mathcal{X}^t)$ for $u \in [t,T]$ are given by  \eqref{xdynam}
and thus \eqref{zdyname} is satisfied by  $\wh Y$ and $\wh Z$ with the terminal condition $\wt Y_T=x$.
\end{proof}

Although BSDEs \eqref{zdynam} and \eqref{zdyname} have the same shape, the features of
their solutions heavily depend on a specification of the processes $X^k$ and $\beta^{k,l}$.
The next result shows that the gained value and the ex-dividend price coincide when the valuation problem is local,
so that the corresponding BSDEs are classical.
In contrast, this property will typically fail to hold when a valuation problem is global, so that \eqref{zdynam} becomes
a generalized BSDE. In that case, equation \eqref{zdyname} needs to be complemented by additional conditions regarding
the processes $X^k$ and $\beta^{k,l}$.

\begin{proposition} \label{prop5.2}
Under the assumptions of Proposition \ref{prop5.1}, if a valuation problem is local, then for any contract $\cC \in \sC$
the hedger's gained value and the hedger's ex-dividend price satisfy $p^e_t(x,\cC)= p^g_t(x,\cC^t)$ for all $t \in [0,T]$.
\end{proposition}

\begin{proof}
On the one hand, under the postulate of uniqueness of solutions to BSDE \eqref{zdynam} (and thus also to BSDE \eqref{zdyname}),
the equality $\wh Y_t=\bar Y_t$ is manifestly satisfied for all $t \in [0,T]$. On the other hand, from Definition \ref{defprice3}, we obtain the equality $x_t(x)+p^e_t(x,\cC^t)=B^l_t \bar Y_t$, which in turn yields $p^e_t(x,\cC^t)= B^l_t ( \bar Y_t - x)$.
Since $\wh Y_t=\wt Y_t$, we conclude that the gained value $p^g_t(x,\cC)=B^l_t ( \wh Y_t - x)$ and the ex-dividend price $p^e_t(x,\cC^t)$
coincide for all $t \in [0,T]$.
\end{proof}

The property of a local valuation problem established in Proposition \ref{prop5.2} is fairly general: its validity hinges on the existence and uniqueness of a solution to a common BSDE for the gained value and the ex-dividend price. This should be contrasted with the case of the global valuation problem where the equality $p^g_t(x,\cC)=p^e_t(x,\cC^t)$ is always satisfied for $t=0$, but it is not likely to hold for any $t > 0$.

\subsection{BSDE for the CCR Price} \label{sec6.3}

We now address the question raised in Section \ref{sec2.8.2}: can we disentangle the counterparty risk-free valuation of a
credit risky contract from the CRR valuation? Although this is true in the linear setup where the price
additivity is known to hold, the answer to this question is unlikely to be positive within a nonlinear framework. On the one hand,
according to Proposition \ref{pro1}, the counterparty risky contract $(A^\sharp ,\cX)$ admits the following decomposition
\begin{equation} \label{ccrxx}
(A^\sharp ,\cX )=(A,\cX )+(A^{\rm CCR},0),
\end{equation}
where the first component is not subject to the counterparty credit risk (although it may include the margin account) and thus
it is referred to as the {\it counterparty risk-free} contract, whereas the second component is concerned exclusively with the CCR  (see Definition \ref{expo} for the specification of the CCR cash flow $A^{\rm CCR}$). On the other hand, however, in a nonlinear framework,
the price of the full contract  $(A^\sharp ,\cX )$ is unlikely to be equal to the sum of prices of its components appearing in the additive decomposition of the full contract.

To examine this problem more closely, let us assume that the underlying market model is sufficiently rich to allow
for replication of the full contract $(A^\sharp ,\cX )$, as well as for replication of its two components
$(A,\cX )$ and $(A^{\rm CCR},0)$. Of course, one can alternatively focus on the decomposition
$(A^\sharp ,\cX )=(A,0)+(A^{\rm CCR},\cX )$ in which the trading adjustments (in particular, the margin
account) are assumed to affect the CCR part, rather than the counterparty risk-free contract $(A,0)$. The choice of a decomposition
should be motivated by practical considerations; one may argue that collateralization is nowadays a standard covenant in most
contracts, not necessarily directly related to the actual level of exposure to the counterparty credit risk in a given contract.

If we denote by $\tau^h$ and $\tau^c$ the default times of the hedger and the counterparty, respectively,
then $\tau=\tau^h \wedge \tau^c$ is the moment of the first default and thus the effective maturity of
$(A^\sharp ,\cX )$ and $(A^{\rm CCR},0)$ is the random time $\whtau=\tau \wedge T$. For the counterparty risk-free contract $(A,\cX )$,
it is convenient to formally assume that its maturity date equals $T$, since this component of the full contract is not exposed to
the default risk.

By a minor extension of Lemma \ref{lemm5.1}, we obtain the following BSDE for the full contract $(A^\sharp ,\cX )$
\begin{equation}\label{dynccr1}
\begin{aligned}
d\wh Y_t=&  \sum_{i=1}^{d} \wh Z_t^i\,d\wh S^{i,cld}_t  - (B^{0,b}_t)^{-1} \bigg( \wh Y_t B^{0,l}_t+\sum_{k=1}^n \alpha^k_t X^k_t \bigg)^- dB_t^{0,b,l} \\
&+(B^{0,l}_t)^{-1}\, dA^{\sharp}_t - \sum_{k=1}^{n} \wh X_t^k\,d\wt \beta_t^{k,l}+\sum_{k=1}^{n} (1-\alpha_t^k) X^k_t\,d (B^{0,l}_t)^{-1}
\end{aligned}
\end{equation}
with the terminal condition $\wh Y_{\whtau}=x$. Let $x=x_1+x_2$ be an arbitrary split of the hedger's endowment. Then we obtain
the following BSDE corresponding to the counterparty risk-free contract $(A,\cX )$
\begin{equation}\label{dynccr2}
\begin{aligned}
d\wh Y^1_t=&  \sum_{i=1}^{d} \wh Z_t^{1,i}\,d\wh S^{i,cld}_t  - (B^{0,b}_t)^{-1} \bigg( \wh Y^1_t B^{0,l}_t+\sum_{k=1}^n \alpha^k_t X^k_t \bigg)^- dB_t^{0,b,l}  \\ &+(B^{0,l}_t)^{-1}\, dA_t - \sum_{k=1}^{n} \wh X_t^k\,d\wt \beta_t^{k,l}+\sum_{k=1}^{n} (1-\alpha_t^k) X^k_t\,d (B^{0,l}_t)^{-1}
\end{aligned}
\end{equation}
with $\wh Y^1_T=x_1$. The BSDE associated with the CRR component $(A^{\rm CCR},0)$ reads
\begin{equation} \label{dynccr3}
d\wh Y^2_t=\sum_{i=1}^{d}\wh Z_t^{2,i}\,d\wh S^{i,cld}_t-(B^{0,b}_t)^{-1} \big( \wh Y^2_t B^{0,l}_t\big)^-dB_t^{0,b,l}+(B^{0,l}_t)^{-1}\, dA^{\rm CRR}_t
\end{equation}
with $\wh Y^2_{\whtau}=x_2$. If the initial endowment $x=0$, then we may take $x_1$ and $x_2$ to be null as well.

The question formulated at the beginning of this section can be restated as follows: under which conditions the equality $\wh Y_0= \wh Y^1_0+\wh Y^2_0$ holds for solutions to BSDEs  \eqref{dynccr1}, \eqref{dynccr2} and \eqref{dynccr3}, so that the three replication costs satisfy the following equality
\[
p^r_0(x, A^\sharp ,\cX )=p^r_0(x_1,A,\cX)+p^r_0(x_2,A^{\rm CCR},0),
\]
which formally corresponds to decomposition \eqref{ccrxx} of the full contract and the split $x=x_1+x_2$ of the hedger's initial endowment? Since this equality is unlikely to be satisfied (even when $x=x_1=x_2=0$, as it was implicitly assumed in most existing papers
on the nonlinear approach to credit risk modeling), one could ask, more generally, whether the quantities $\wh Y_0$ and $\wh Y^1_0+\wh Y^2_0$ are close to each other, so that some approximate equality is satisfied by the replication costs.
Of course, an analogous question can also be formulated for the corresponding replicating strategies.

One needs first to address the issues of regularity and completeness of market models with default times. To this end,
one may employ the existence and uniqueness results, as well as the strict comparison theorems,
obtained for BSDEs with jumps generated by the occurrence of random times. BSDEs of this form are relatively uncommon in the
existing literature on the theory of BSDEs, but they were studied in papers by Peng and Xu \cite{PX2009} and Quenez and Sulem \cite{QS2013}. Of course, to be in a position to use results established in those papers, one would need to explicitly specify the price dynamics for non-defaultable risky assets $S^1, \dots , S^{d-2}$ (usually, they are supposed to be driven by a multidimensional Brownian motion), as well as the manner in which default times (thus also the prices of defaultable assets $S^{d-1}$ and $S^d$) are defined. The latter issue is addressed in Peng and Xu \cite{PX2009} or Quenez and Sulem \cite{QS2013} through the so-called intensity-based approach, which was previously extensively studied in the credit risk literature. Furthermore, it would be convenient to assume that the cash and funding accounts, as well as remuneration processes, have absolutely continuous sample paths, so that BSDEs could be represented in the following generic form
\[
dY_t=- g(t,Z_t,Y_t)\,dt+\sum_{i=1}^{d-2}Z_t^i\,dW^{i}_t+\sum_{i=d-1}^{d}Z_t^i\,dM^{i}_t+d\bar A_t,
\]
where $M^1$ and $M^2$ are purely discontinuous $\bG$-martingales associated with the price processes $S^{d-1}$ and $S^d$, respectively, and $\bar A$ is a fixed process. Note that the generator $g$ can be obtained from \eqref{dynccr1}, \eqref{dynccr2} and \eqref{dynccr3} by straightforward computations.

From the financial perspective, to ensure the completeness of the market model at hand, one would need to postulate that some defaultable securities (typically, defaultable bonds issued by the two parties or credit default swaps) are among primary traded assets.
Finally, it is also necessary to explicitly specify the closeout valuation process $Q$ (see Remark \ref{remccr}) and the collateral process $C$. When dealing with a local valuation problem, the most natural theoretical choice (albeit not necessarily easy to implement in practice) would be to set (see Proposition \ref{prop5.2})
$$
Q_t := p^e_t(x_1,\cC)=p^g_t(x_1,\cC^t), \quad C_t := p^e_t(x, A^\sharp ,\cX)=p^g_t(x,(A^\sharp)^t,\cX^t )
$$
although the latter convention of the {\it endogenous collateral} is slightly cumbersome to handle, even when dealing with BSDEs driven by a multidimensional continuous martingale (see Nie and Rutkowski \cite{NR2016b}). Note also that it would require to replace $C_{\tau }$ by $C_{\tau -}$ when specifying the closeout payoff $\mathfrak{K}$ (hence also the process $A^\sharp $) in Section \ref{sec2.8.1}. For technical problems for BSDEs with jumps arising in this context and related to the classical {\it immersion hypothesis} and the way in which they can be resolved, the interested reader is referred to recent papers by Cr\'epey and Song \cite{CS2015b,CS2015a}.

 Within the framework of a linear model of credit risk, the issue of market completeness and various methods for replication were studied in several works (see, in particular, Bielecki et al. \cite{BJR2004,BJR2006,BJR2008}). In contrast, only a few papers devoted to nonlinear models of credit risk are available. More recently, Cr\'epey \cite{C2015a,C2015b}, Dumitrescu et al. \cite{DQS2017} and Bichuch et al. \cite{BCS2017} used BSDEs with jumps to solve the valuation and hedging problems for derivative contracts exposed to the counterparty credit risk. In Bichuch et al. \cite{BCS2017} and Dumitrescu et al. \cite{DQS2017}, the authors focus on valuation of the full contract, whereas Cr\'epey \cite{C2015a,C2015b} examines the problem of the approximate additivity for the credit valuation adjustments.

\section{Nonlinear Valuation Versus Market Practice} \label{sec7}

Although the goal of this paper is to formulate questions and give preliminary answers regarding the most fundamental issues pertinent to the theory of {\bf nonlinear} arbitrage-free pricing, a few remarks concerning the current market practice and its relationship to theoretical results on nonlinear pricing could be appreciated by the reader. We also briefly describe some related recent papers where the issue of the so-called valuation adjustments was examined in both linear and nonlinear setups.

According to the prevailing practice, the full price for a counterparty risky contract is obtained by combining, at least implicitly, the so-called {\it clean price} of the basic contract $(A,0)$ with various {\it valuation adjustments}. The clean price and the corresponding hedge are computed by the trading desk using the classical linear approach under a manifestly unrealistic, but obviously very convenient, assumption that a proxy for the unique risk-free rate is available for funding of all trading activities and the counterparty credit risk is ignored. It is thus clear that, from the theoretical point of view, the clean price can be computed as a solution to a linear BSDE, as in the classical arbitrage-free pricing theory.  In contrast, various valuation adjustments are determined by the dedicated CVA desk in such a way that they account for all other features of a contract and trading conditions, such as: differential funding costs,  the presence of the margin account, the counterparty credit risk, regulatory requirements, etc.. This means that, according to practical approach adopted by most banks, the full price of a new deal is implicitly represented as follows
\begin{equation}
\begin{aligned} \label{mmjj}
\textrm{full price}\,(A^\sharp,\cX )&:=\, \textrm{linear price}\,(A,0)+\textrm{possibly nonlinear price}\,(A^{\rm CCR},\cX ) \\
 &=\,\textrm{\it clean price }+\, \textrm{{\it total valuation adjustment} (XVA),}
\end{aligned}
\end{equation}
where the clean price is given by a solution to a linear BSDE and the total valuation adjustment (denoted as XVA) is determined by solving either a linear or a nonlinear BSDE. Of course, if all three terms appearing in \eqref{mmjj} (that is: the full price, the clean price, and the total valuation adjustment) are given by solutions to particular linear BSDEs, then it is possible to argue that decomposition \eqref{mmjj} can be formally justified. However, if the valuation adjustment (and thus also the full price) is given by a solution to a nonlinear BSDE, which is the case of our primary interest, then the two terms appearing in the right-hand side in \eqref{mmjj} cannot be computed separately and subsequently aggregated to obtain the full price. This observation is valid, in general, despite the fact that the clean price is always computed through a solution to linear BSDE or, equivalently, a suitable version of the risk-neutral valuation formula, so that it enjoys the additivity property across several (uncollateralized and non-defaultable) deals. Therefore, in our opinion, the introduction of the concept of the clean price, although convenient in practice since it refers to the pre-crisis experience and facilitates calibration of commonly used models for the underlying securities, may further complicate the theoretical problem of searching for the full price of a contract and the corresponding hedging strategy when working in a nonlinear setup.

Let us illustrate the issue of non-additivity by focusing on a specific nonlinear setup. We stress that by the {\it non-additivity} of pricing we mean here the property that the clean price and the total valuation adjustment cannot be computed separately by splitting the cash flows of a contract and perhaps even using a different model to deal with each of the two (or more) components. Brigo et al. \cite{BB2018} have recently shown that the risk-neutral valuation approach with adjusted cash flows based on a proxy for the risk-free interest rate, which was introduced and studied by Pallavicini et al. \cite{PPB2012a,PPB2012b},  can be formally supported using the replication-based approach in which a proxy for the risk-free interest rate can be chosen in a completely arbitrary way. Indeed, it is possible to use any $\bG$-adapted and suitably integrable process $\alpha $ to play the role of a proxy for the risk-free interest rate, since the financial interpretation (if any) of this process is irrelevant for the derivation of decomposition \eqref{eqv3mc}. It was proven in Proposition 3.2 in \cite{BB2018}, under the assumption of null initial endowment (so that $x=x_t(x)=0$ for all $t\in [0,T]$), that the ex-dividend hedger's price of an attainable collateralized contract $\cC = (A^\sharp,C)$ equals, on the event $\{ t < \tau \}$ for every $t \in [0,T]$,
 \begin{align} \label{eqv3mc}
&p^e_t (0,\cC^t)= p^{e,\alpha }_t(A)+B^{\alpha}_t\,\mathbb{E}_{{\mathbb Q}^{\alpha}} \Big( \I_{\{ \tau\leq T\}} \big(\I_{\{\tau^c<\tau^h \}} (1-R_c) \Upsilon^+ -\I_{\{\tau^h<\tau^c\}}(1-R_h)\Upsilon^- \big)\,\Big|\,{\cal G}_t \Big) \nonumber \\ &
+ B^{\alpha}_t\,\mathbb{E}_{{\mathbb Q}^{\alpha}} \bigg(\int_t^{\tau \wedge T}(\alpha _u-\bar{f}_u)F_u(B^{\alpha}_u)^{-1}\,du+\sum_{i=1}^d \int_t^{\tau \wedge T}(\alpha_u-\bar{h}^i_u)F^i_u (B^{\alpha}_u)^{-1}\,du \,\Big|\,{\cal G}_t \bigg) \\ &
+ B^{\alpha}_t\,\mathbb{E}_{{\mathbb Q}^{\alpha}}\bigg(\int_t^{\tau \wedge T}(\bar{c}_u-\alpha_u)C_u(B^{\alpha}_u)^{-1}\,du\,\Big|\,{\cal G}_t\bigg), \nonumber
\end{align}
where we denote $F_t =\psi^{0,l}_t B^{0,l}_t+\psi^{0,b}_t B^{0,b}_t,\, F^i_t =\xi^i_t S^i_t$, and
$$
\bar{f}_t := f^l_t \I_{\{F_t \geq 0\}}+f^b_t  \I_{\{F_t<0\}}, \
\bar{h}^i_t := h^{i,l}_t \I_{\{F^i_t\geq 0\}}+h^{i,b}_t  \I_{\{F^i_t<0\}}, \
\bar{c}_t := c^l_t \I_{\{C_t<0\}}+c^b_t  \I_{\{C_t\geq 0\}},
$$
where $c^l$ (resp. $c^b$) is the remuneration rate for the cash collateral pledged (resp. received) by the hedger.
Furthermore, $p^{e,\alpha}_t(A)$ is the ex-dividend clean price and the probability measure ${\mathbb Q}^{\alpha}$ is such that the processes $S^i(B^{\alpha})^{-1},\, i=1,2, \dots , d$, are  ${\mathbb Q}^{\alpha}$-martingales. For more information on the concept of a ``martingale measure'' in a nonlinear model, see Section 3.2.1 in \cite{BB2018}. Equality \eqref{eqv3mc} leads to the following formal additive decomposition of the ex-dividend price for $(A^\sharp,\cX )$ into its clean price and several complementary valuation adjustments (for their interpretation, see Section 3.2.2 in \cite{BB2018}), which can be aggregated into a single {\it total valuation adjustment,} denoted as $\XVA_t$, so that we have
\begin{equation}
\begin{aligned} \label{xvaformulan}
p^e_t (0,\cC^t)& =\pi^{e,\alpha }_t(A)+\CVA_t-\DVA_t+\FVA^{\bar{f}}_t+\sum_{i=1}^d \FVA^{\bar{h}^i}_t+\LVA_t \\
&= \pi^{e,\alpha }_t(A) + \XVA_t .
\end{aligned}
\end{equation}
We stress that \eqref{xvaformulan} is true for any choice as a proxy $\alpha $ for the risk-free interest rate,
which further supports the view that the clean price is an abstract concept dissociated from the actual trading and
the total valuation adjustment is a necessary mechanism needed to bring it back to reality.
More importantly, terms appearing in the right-hand side in \eqref{xvaformulan} are intertwined so that
various valuation adjustments cannot be computed without the prior knowledge of the hedging strategy for the full contract.
We thus conclude that the additivity and separation of adjustments, which is visibly suggested by the shape of equality \eqref{xvaformulan}, is in fact illusory, unless the underlying trading model has fully linear features so that it is possible to use the theory of linear BSDEs in order to justify separation. Obviously, this does not mean that separation cannot hold in some models with certain nonlinear features but this should be an exceptional situation, rather than the rule.

As a concrete example of an explicit application of nonlinear pricing theory, we may quote the recent paper by Bichuch et al. \cite{BCS2017} who provide a thorough examination of valuation of path-independent European claims in an extension of the classical
Black-Scholes model to differential funding rates and counterparty credit risk (for another example of pricing under asymmetric borrowing and lending rates, see Brigo and Pallavicini \cite{BP2014}).  In \cite{BCS2017}, the authors first verify the no-arbitrage property  of their trading model with respect to the null contract in the case of a nonnegative initial endowment $x$. Subsequently, they apply the BSDE approach to unilateral valuation of a collateralized European claim with bilateral default risk. The closeout payoff is specified in reference to the third party valuation, which is based on a single risk-free rate (not available to the two parties) and thus it is given by the standard Black-Scholes model. It is important to stress that the total valuation adjustment for the hedger (or the counterparty) is not computed in \cite{BP2014} as a separate quantity, but it is instead {\bf defined} as the difference between the full unilateral price and the Black-Scholes price (see Definition 4.8 in \cite{BCS2017}), which is hence supposed to play the role of the clean price. Using our notation, the definition of the total valuation adjustment adopted in \cite{BCS2017} reads $\XVA_t := p^e_t (x,\cC^t) - \pi^{e,\alpha }_t(A)$.  It is thus clear that  Bichuch et al.  \cite{BCS2017} do not advocate the practical approach, where the clean price and valuations adjustment are supposed to be first independently computed by trading and CVA desks and subsequently aggregated into the full price. It is also observed  in \cite{BCS2017} that the total valuation adjustments are equal, and thus unilateral prices collapse to a single full price, when the pricing BSDE is linear; otherwise, unilateral full prices computed by the two counterparties, who are supposed to use identical trading model, are likely to differ. Obviously, it is not our intention to suggest that the practice where separate desks are independently dealing with components of a contract and then using the aggregate number as a plausible candidate for the ``full price'' is wrong and thus should be discontinued. We have only argued that this pragmatic approach, which can be justified within the framework of a linear market (see, for instance, Burgard and Kjaer \cite{BK2011,BK2013}, Fujii and Takahashi \cite{FT2013} or Kenyon and Green \cite{KG2014a,KG2014b}) is unlikely to result in a mathematically sound arbitrage-free pricing of derivatives in a nonlinear setup, where the introduction of the concept of the clean price is no longer advantageous.

Let us finally mention the important issue of {\it netting} of outstanding deals between counterparties, which means that, in principle, every new deal should be valued not in isolation, but rather as a new component added to the portfolio of existing contracts. Needless to say, this issue is highly challenging, both in theory and practice, and thus it is left for future work. Last but not least, it should be acknowledged that the valuation of derivatives based on arbitrage-free replication (or superhedging) should not be seen as the most realistic pricing approach, but rather a mathematical idealization of a much more complex situation, and thus other pricing paradigms should also be examined. The interested reader is referred to  Kenyon and Green \cite{KG2013} for a discussion of a regulatory-compliant derivatives pricing and to Albanese and Cr\'epey \cite{AC2017} for a novel balance-sheet approach to XVA with the special emphasis on KVA (capital value adjustment) computations.

\newpage

\section*{Acknowledgments}
The research of I.~Cialenco and M.~Rutkowski was supported by the DVC Research Bridging Support Grant {\it BSDEs Approach to Models with Funding Costs.} Part of the research was completed while I.~Cialenco and M.~Rutkowski were visiting the Institute for Pure and Applied Mathematics (IPAM) at UCLA, which is funded by the National Science Foundation. We would also like to thank the anonymous referees and St\'ephane Cr\'epey for their insightful and helpful comments and suggestions, which helped us greatly to improve the final manuscript.


\end{document}